\documentclass[12pt, DIV10,a4paper]{article}

\usepackage{color}
\usepackage{float}
\usepackage[bf,small]{caption2}
\usepackage{comment}
\usepackage{amsmath}
\usepackage{bigints}
\usepackage{rotating, booktabs}
\usepackage{amsopn}
\usepackage{amsfonts}
\usepackage{amsthm}
\usepackage{amssymb}
\usepackage{amsbsy}
\usepackage{gensymb}
\usepackage{multirow}
\usepackage{titling}
\usepackage{natbib}
\usepackage{tocbibind}
\usepackage[a4paper,colorlinks,breaklinks,bookmarksopen,bookmarksnumbered]{hyperref}
\usepackage{epsfig,psfrag}
\usepackage{graphicx}
\usepackage{amsmath}
\usepackage{epstopdf}
\usepackage{tabularx}
\usepackage{subfigure}
\usepackage[mathscr]{euscript}
\usepackage[margin=1in]{geometry}
\usepackage{xr-hyper}

\usepackage{amsfonts} 
\usepackage{calc}
\usepackage{titlesec}

\newcommand{\ueq}[1][]{%
  \if\relax\detokenize{#1}\relax
    \sbox0{$\underbrace{=}_{}$}%
    \mathrel{\mathmakebox[\wd0]{=}}
  \else
    \mathrel{\underbrace{=}_{\mathclap{#1}}}
  \fi}

\newcommand{\bzero}{\boldsymbol{0}}

\newcommand {\ctn}{\cite}

\usepackage{url}
\renewcommand{\t}{\ensuremath{\theta}}

\newcommand{\btheta}{\boldsymbol{\theta}}

\newcommand{\bxi}{\boldsymbol{\xi}}

\newcommand{\bgamma}{\boldsymbol{\gamma}}
\newcommand{\bSigma}{\boldsymbol{\Sigma}}

\newcommand{\bpi}{\boldsymbol{\pi}}

\newcommand{\bmu}{\boldsymbol{\mu}}

\newcommand{\bvarphi}{\boldsymbol{\varphi}}
\newcommand{\bvartheta}{\boldsymbol{\vartheta}}

\newcommand{\bB}{\boldsymbol{B}}

\newcommand{\bG}{\boldsymbol{G}}

\newcommand{\bA}{\boldsymbol{A}}

\newcommand{\bU}{\boldsymbol{U}}

\newcommand{\bX}{\boldsymbol{X}}
\newcommand{\by}{\boldsymbol{y}}
\newcommand{\bY}{\boldsymbol{Y}}

\newcommand{\bz}{\boldsymbol{z}}

\newtheorem{theorem}{Theorem}

\newtheorem{remark}[theorem]{Remark}

\renewcommand{\t}{\ensuremath{\theta}}

\newcommand{\topline}{\hrule height 1pt width \textwidth \vspace*{2pt}}
\newcommand{\botline}{\vspace*{2pt}\hrule height 1pt width \textwidth \vspace*{4pt}}
\newtheorem{algo}{Algorithm} 

\begin{document}

\title{\vspace{-0.8in}
\textbf{IID Sampling from Posterior Dirichlet Process Mixtures}}
\author{Sourabh Bhattacharya\thanks{
Sourabh Bhattacharya is an Associate Professor in Interdisciplinary Statistical Research Unit, Indian Statistical
Institute, 203, B. T. Road, Kolkata 700108.
Corresponding e-mail: sourabh@isical.ac.in.}}
\date{\vspace{-0.5in}}
\maketitle%
	
\begin{abstract}

The influence of Dirichlet process mixture is ubiquitous in the Bayesian nonparametrics literature. But sampling from its posterior distribution
remains a challenge, despite the advent of various Markov chain Monte Carlo methods. The primary challenge is the infinite-dimensional setup,
and even if the infinite-dimensional random measure is integrated out, high-dimensionality and discreteness still remain difficult issues to deal with.

In this article, exploiting the key ideas proposed in \ctn{Bhatta21a}, we propose a novel methodology for drawing $iid$ realizations from
posteriors of Dirichlet process mixtures. We focus in particular on the more general and flexible model of \ctn{Bhattacharya08}, so that the methods
developed here are simply applicable to the traditional Dirichlet process mixture.

We illustrate our ideas on the well-known enzyme, acidity and the galaxy datasets, which are usually considered benchmark datasets for mixture applications.
Generating $10,000$ $iid$ realizations from the Dirichlet process mixture posterior of \ctn{Bhattacharya08} given these datasets took $19$ minutes, $8$ minutes 
and $5$ minutes, respectively, in our parallel implementation.
\\[2mm]

\end{abstract}
	
\pagebreak

\section{Introduction}
\label{sec:introduction}
The Bayesian nonparametric literature is heavily dominated by Dirichlet process (DP) mixtures, which go back to \ctn{Antoniak74} and \ctn{Ferguson83}.
The basic premise is given by the following setup: for observed data $y_i$: $i=1,\ldots,n$, the conditional distribution of $y_i$ given parameters
$\btheta_i$ is $[y_i|\btheta_i]\sim f(\cdot|\btheta_i)$, where $f(\cdot|\btheta)$ is a parametric distribution with parameters $\btheta$, and $\btheta_i\sim G$ independently,
where $G$ is a random distribution on which some appropriate prior must be assigned. The DP mixture model considers the following prior for $G$:
$G\sim DP(\alpha G_0)$, the Dirichlet process prior introduced in \ctn{Ferguson73}; here $\alpha>0$ is a scale parameter and $G_0$ is the base (expected)
distribution of $G$. Thus, conditionally on $G$, the distribution of $y_i$ is a mixture of $f(\cdot|\btheta)$ over the distribution $G$ of $\btheta$. 
Usually a prior is placed on the scale parameter $\alpha$. 

Although DP mixtures have already seen applications in almost all areas of statistics, its journey perhaps began with the recognition of its versatility
with respect to clustering, nonparametric regression and nonparametric density estimation. Of course, the beginning of the computer era in the $1990$s  
played a very significant role in the development of the computational aspects of posterior DP mixtures. \ctn{Escobar94}, \ctn{Escobar95}, \ctn{west94}, \ctn{MacEachern94},
\ctn{Muller96}, \ctn{MacEachern98}, etc. seem to recognize
the practical and computational aspects of such models and developed various Gibbs sampling algorithms based on a P\'{o}lya urn scheme obtained after 
integrating out the infinite-dimensional random measure $G$. \ctn{Neal00} provided a comprehensive overview of the various Markov chain Monte Carlo (MCMC) algorithms used for  
sampling from posterior DP mixtures, and also provided algorithms for non-conjugate setups, that is, when $f(\cdot|\btheta)$ and $G_0(\btheta)$ are non-conjugate.
\ctn{Green01}, \ctn{Jain04} and \ctn{Jain07} propose split-merge moves embedded in reversible jump (\ctn{Green95}, \ctn{Richardson97}) 
and Metropolis-Hastings procedures to implement DP mixtures, in conjugate and non-conjugate setups.

\ctn{Ishwaran01} (see also \ctn{Ishwaran00}) proposed a block Gibbs sampling algorithm when $G$ is retained in the model; their key idea is to truncate
$G$ to a finite-dimensional random measure such that the latter is almost indistinguishable from the original random measure. 
On the other hand, \ctn{Papas08} proposed a retrospective MCMC method which does not require truncation of $G$; an alternative method based on slice sampling
is proposed by \ctn{Walker07}.

All the existing MCMC sampling methods for DP mixtures have their advantages and disadvantages with respect to mixing and implementation time, and it is difficult
to single out any MCMC method that is guaranteed to outperform the others in all situations.
The ideal scenario, although it might seem ``too ambitious" to the statistical and probabilistic community, is to devise an $iid$ sampling procedure.
Indeed, our objective in this article is to propose a novel methodology for generating $iid$ realizations from the posterior of DP mixtures. We specifically
focus on the much more flexible and efficient model proposed by \ctn{Bhattacharya08}, which includes the traditional DP mixture as a special case. Hence, although we develop
the $iid$ sampling method with respect to \ctn{Bhattacharya08}, it is simply applicable to the traditional DP mixture. Our idea is to first truncate $G$ to render it 
finite-dimensional, but such that the truncated version is practically indistinguishable from the original one. Indeed, we obtain an upper bound for the $L_1$-distance 
between the predictive distributions of the original and truncated versions which is very significantly smaller than the bound obtained by \ctn{Ishwaran01}
for the traditional DP mixture.
Such a bound ensures that the posterior realizations under the original random measure and the truncated one, are identical in practice.
The key to our highly efficient upper bound is the bounded number of components of the mixture model for the observations, which are mixed with respect to the DP.

Once such truncation is established, we invoke the general $iid$ sampling strategy on finite-dimensional Euclidean spaces proposed by \ctn{Bhatta21a}. 
In a nutshell, the idea is to create an infinite sequence of closed, concentric ellipsoids, representing the target distribution as an infinite mixture
on the ellipsoids and the annuli (regions between successive concentric ellipsoids), drawing a mixture component with the appropriate probability
and finally simulating perfectly from the mixture component using a novel strategy. In our DP context, although the parameters associated with the truncated
random measure can be represented in a finite-dimensional Euclidean space, the parameters of the mixture distribution of the observations coincide with each other
with positive probabilities, and hence the method of \ctn{Bhatta21a} can not be directly applied here. We thus extend his procedure by including the truncated
random measure in the proposal associated with perfect sampling strategy, so that once its parameters are simulated aided by a suitable diffeomorphic
transformation for efficiency, the rest of the parameters are simply drawn from the truncated measure,
in a way that the entire procedure of $iid$ sampling remains ``perfect".

We apply our $iid$ sampling method to three well-known datasets, namely, the enzyme, acidity and galaxy data, which are are usually considered to be
benchmarks for mixture applications. Generation of $10,000$ $iid$ realizations from the posterior of \ctn{Bhattacharya08} for these datasets took
$19$ minutes, $8$ minutes and $5$ minutes, respectively, with parallel implementation on $80$ cores. The resultant Bayesian inferences turned out to be
very encouraging.

The rest of our article is organized as follows. In Section \ref{sec:dp_bounded} we begin with a brief description of the DP mixture model of \ctn{Bhattacharya08}.
The $iid$ sampling idea for such DP mixture is detailed in Section \ref{sec:idea}. In Section \ref{sec:applications} we provide details on the application of our $iid$ 
sampling procedure to the three benchmark datasets.
We summarize our ideas and make concluding remarks in Section \ref{sec:conclusion}.

\section{The DP mixture with bounded number of components for the observational mixture model}
\label{sec:dp_bounded}

\subsection{Model description}
\label{subsec:model}
Letting $\by=(y_1,\ldots,y_n)$ denote the data of size $n$, a slightly extended version of the DP mixture model of \ctn{Bhattacharya08} is as follows:
\begin{align}
	[\by|\bxi,\bpi]&\stackrel{iid}{\sim}\sum_{j=1}^M\pi_jf(\cdot|\xi_j);\label{eq:dp1}\\
	[\xi_1,\ldots,\xi_M|G]&\stackrel{iid}{\sim} G;\label{eq:dp2}\\
	G&\sim DP\left(\alpha G_0\right);\label{eq:dp3}\\
	\pi_j&=\frac{\exp\left(\psi_j\right)}{\sum_{k=1}^M\exp\left(\psi_k\right)};~j=1,\ldots,M;\label{eq:dp4}\\
	\psi_j&\stackrel{iid}{\sim}f_{\psi},\label{eq:dp5}
\end{align}
where $\bxi=(\xi_1,\ldots,\xi_M)$, $\bpi=(\pi_1,\ldots,\pi_M)$ and $f_{\psi}$ denotes any appropriate prior distribution for the $\psi_j$.


Note that (\ref{eq:dp1}) shows that the model for the individual observations is a mixture with a maximum of $M$ components.
Also observe that under the above model, given $G$, 
for any value of $M$, the prior predictive distribution is given by
\begin{eqnarray}
	f_{G}(y)&=& \sum_{i=1}^M\pi_i \int f(y\vert\xi_i)\prod_{j=1}^MdG(\xi_j)\nonumber\\
&=&\sum_{i=1}^M\pi_i\int f(y\vert\xi_i)dG(\xi_i)\nonumber\\
&=&\int f(y\vert\xi)dG(\xi),\notag
\end{eqnarray}
so that the marginal distribution of any data point given $G$ is the same as that of the traditional DP mixture. However, given $G$, $y_1,\ldots,y_n$ are
not independent, as their joint distribution conditional on $G$, shows below: 
\begin{equation*}
[\by|G,\bpi]=\int \left\{\prod_{i=1}^n\left[\sum_{j=1}^M\pi_jf(y_i|\xi_j)\right]\right\}\prod_{j=1}^MdG(\xi_j).
\end{equation*}
Thus, the DP mixture model of \ctn{Bhattacharya08} is very significantly different from the traditional DP mixture, and is perhaps much more realistic in terms
of the dependence structure. Further note that if $M=n$, $\pi_j=\frac{1}{M}$ for $j=1,\ldots,M$ and for each $i$, $y_i$ is set to come from $f(\cdot|\xi_i)$,
then the above model reduces to the traditional DP model. Thus, the traditional DP model is a special case of \ctn{Bhattacharya08}. Numerous theoretical, asymptotical,
computational and application-wise advantages of DP mixture of \ctn{Bhattacharya08} over the traditional DP mixture are noted in 
\ctn{Bhattacharya08}, \ctn{Sabya11}, \ctn{Sabya12b}, \ctn{Sabya21}.

\subsection{Truncation of the infinite-dimensional random measure}
\label{subsec:truncation}

It holds almost surely (see, for example, \ctn{Sethuraman94}), that 
\begin{equation}
G=\sum_{i=1}^{\infty}w_i\delta_{\phi_i},
	\label{eq:dp6}
\end{equation}
where $w_1=V_1$ and for $i=2,3,\ldots$, $w_i=V_i\prod_{j<i}(1-V_j)$, $\phi_i\stackrel{iid}{\sim}G_0$ and $V_i\stackrel{iid}{\sim}Beta(1,\alpha)$. 

As in \ctn{Ishwaran01} (see also \ctn{Ishwaran00}), we consider the following truncation of (\ref{eq:dp6}): $w_1=V_1$ and $w_i=V_i\prod_{j<i}(1-V_j)$, 
for $j=2,\ldots,N$. We set $V_N=1$
so that $\sum_{i=1}^Nw_i=1$. Let $G_N$ denote the truncated probability measure corresponding to (\ref{eq:dp6}) with $N$ summands.
That is,
\begin{equation}
G_N=\sum_{i=1}^{N}w_i\delta_{\phi_i},
	\label{eq:dp_trunc}
\end{equation}
Let 
$\bvarphi_N=(\alpha,\psi_1,\ldots,\psi_M,\phi_1,\ldots,\phi_N,V_1,\ldots,V_{N-1})$ and 
$\bvarphi=(\alpha,\psi_1,\ldots,\psi_M,\phi_1,\phi_2,\ldots,V_1,V_2,\ldots)$.

Let $\bz=(z_1,\ldots,z_n)$ denote the allocation variables correspondng to $\by$, that is, for $i=1,\ldots,n$, and $j=1,\ldots,M$, 
$z_i=j$ indicates that $[y_i|z_i=j,\bxi]\sim f(\cdot|\xi_j)$. The probability of the event $z_i=j$ is given by $[z_i=j]=\pi_j$. With these,
consider the marginal distribution of $\by$ corresponding to $G_N$ as follows:
\begin{align}
	m_N(\by)&=\sum_{\bz}[\bz]\int\prod_{i=1}^n[y_i|z_i,\bxi][d\bxi|\bvarphi_N][\bvarphi_N]d\bvarphi_N\notag\\
	&=\sum_{\bz}[\bz]\int\prod_{i=1}^n[y_i|z_i,\bxi]\pi_N(d\bxi),
	\label{eq:marg1}
\end{align}
where $\pi_N(\bxi)=\int [\bxi|\bvarphi_N][\bvarphi_N]d\bvarphi_N$ stands for the marginal distribution of $\bxi$ corresponding to $G_N$.
The marginal distribution of $\by$ corresponding to $G$ is given by
\begin{align}
	m_{\infty}(\by)&=\sum_{\bz}[\bz]\int\prod_{i=1}^n[y_i|z_i,\bxi][d\bxi|\bvarphi][\bvarphi]d\bvarphi\notag\\
	&=\sum_{\bz}[\bz]\int\prod_{i=1}^n[y_i|z_i,\bxi]\pi_{\infty}(d\bxi),
	\label{eq:marg2}
\end{align}
where $\pi_{\infty}(\bxi)=\int [\bxi|\bvarphi][\bvarphi]d\bvarphi$ stands for the marginal distribution of $\bxi$ corresponding to $G$.
\begin{theorem}
	\label{theorem:truncation}
	\begin{align}
	\int|m_N(\by)-m_{\infty}(\by)|d\by\leq 2\left[1-E\left\{\left(\sum_{i=1}^{N-1}w_i\right)^M\right\}\right]\approx 4M\exp\left(-(N-1)/\alpha\right).\notag
	\end{align}
\end{theorem}
\begin{proof}
Note that
	\begin{align}
		|m_N(\by)-m_{\infty}(\by)|\leq\sum_{\bz}[\bz]\int\prod_{i=1}^n[y_i|z_i,\bxi]|\pi_N(d\bxi)-\pi_{\infty}(d\bxi)|,\notag\\
	\end{align}
	so that 
	\begin{equation}
		\int|m_N(\by)-m_{\infty}(\by)|d\by\leq 2 D(\pi_N,\pi_{\infty}),
	\end{equation}
	where $D(\pi_N,\pi_{\infty})$ is the total variation distance between the probability measures $\pi_N$ and $\pi_{\infty}$.
	The rest of the proof follows in the similar lines as that of \ctn{Ishwaran00}.
\end{proof}

\begin{remark}
\label{remark:truncation}
The crucial advantage of the upper bound of Theorem \ref{theorem:truncation} is that the bound depends only upon $M$, $N$ and $\alpha$, and not upon $n$, the
sample size. Although $n$ may be very large, $M$ is usually chosen to be much smaller, and hence our upper bound is significantly smaller than the corresponding
upper bound of \ctn{Ishwaran01} in the traditional DP mixture context, given by $4n\exp\left\{-(N-1)/\alpha\right\}$.

To illustrate the differences between the two different upper bounds, note that with $M=30$ and $N=50$, for $\alpha=3$ for instance, 
our upper bound is given by $4M\exp\left\{-(N-1)/\alpha\right\}=9.676\times 10^{-6}$, 
whereas for the traditional DP mixture model, for $n=245$, the size of the enzyme dataset, 
the corresponding upper bound of \ctn{Ishwaran01} is $4n\exp\left\{-(N-1)/\alpha\right\}=7.902\times 10^{-5}$. 
For the sizes $n=155$ and $n=82$ for the acidity and the galaxy datasets, the upper bound for the \ctn{Bhattacharya08} model remains the same for the same $M$, $N$ and $\alpha$,
but for the traditional DP mixture, the upper bounds are $4.999\times 10^{-5}$ and $2.645\times 10^{-5}$, respectively.
Thus, the upper bound for our model is an order of magnitude smaller than for the traditional DP mixture.
\end{remark}

\subsection{Reparameterization}
\label{subsec:reparameterization}
For our convenience, for $i=1,\ldots,N-1$, let us reparameterize $V_i$ as 
$V_i=\frac{\exp\left(\zeta_i\right)}{1+\exp\left(\zeta_i\right)}$ and $\alpha$ as $\alpha=\exp\left(\tilde\alpha\right)$.
Let $\btheta=(\bxi,\bvartheta)$, where  
$\bvartheta=\left(\tilde\alpha,\psi_1,\ldots,\psi_M,\phi_1,\ldots,\phi_N,\zeta_1,\ldots,\zeta_N\right)$.
Then the reparameterized version of the joint posterior, proportional to likelihood times prior becomes
\begin{align}
	\pi(\btheta|\by)\propto	\prod_{i=1}^n\left[\sum_{j=1}^M\pi_jf(y_i|\xi_j)\right]
	\times\prod_{j=1}^M[\psi_j]\times\prod_{j=1}^MG_N(\xi_j)\times\prod_{i=1}^{N}[\phi_i]\times\prod_{i=1}^{N}[\zeta_i]\times[\tilde\alpha].
	\label{eq:dp7}
\end{align}
We shall henceforth consider this reparameterized setup for our purpose.

\section{The $iid$ sampling idea}
\label{sec:idea}

Note that our DP mixture posterior distribution can be represented as
\begin{equation}
	\pi(\btheta|\by)=\sum_{i=1}^{\infty}\pi(\bB_i\times\bA_i|\by)\pi_i(\btheta|\by),
	\label{eq:p1}
\end{equation}
where $\bA_i$ are disjoint compact subsets of $\mathbb R^d$ with $d=2N+M+1$ such that $\cup_{i=1}^{\infty}\bB_i\times\bA_i=\mathbb R^{d+M}$. Here $\bB_i$'s correspond to
$\bxi$ and $\bA_i$'s correspond to $\bvartheta$. In (\ref{eq:p1}),
\begin{equation}
	\pi_i(\btheta|\by)=\frac{\pi(\btheta|\by)}{\pi(\bB_i\times\bA_i|\by)}I_{\bB_i\times\bA_i}(\btheta), 
	\label{eq:p2}
\end{equation}
is the distribution of $\btheta$ restricted on $\bB_i\times\bA_i$; $I_{\bB_i\times\bA_i}$ being the indicator function of $\bB_i\times\bA_i$.
Also, $\pi(\bB_i\times\bA_i|\by)=\int_{\bB_i\times\bA_i}\pi(d\btheta|\by)\geq 0$. Clearly, $\sum_{i=1}^{\infty}\pi(\bB_i\times\bA_i|\by)=1$.

The key idea of generating $iid$ realizations from $\pi(\btheta|\by)$ is to randomly select $\pi_i(\cdot|\by)$ with probability $\pi(\bB_i\times\bA_i|\by)$ 
and then to perfectly simulate from $\pi_i(\cdot|\by)$. 

Note that due to (\ref{eq:dp_trunc}), $\xi_j$'s must take one of the $\phi_i$ values. Hence, the choice of the sets $\bA_i$ determine the sets $\bB_i$.
Hence, it is sufficient to adequately choose $\bA_i$, the method of which we discuss next.

\subsection{Choice of the sets $\bA_i$ and estimation of $\pi(\bB_i\times\bA_i|\by)$}
\label{subsec:choice_sets}
For some appropriate $d$-dimensional vector $\bmu$ and $d\times d$ positive definite scale matrix $\bSigma$,
we shall set $\bA_i=\{\bvartheta:c_{i-1}\leq (\bvartheta-\bmu)^T\bSigma^{-1}(\bvartheta-\bmu)\leq c_i\}$ for $i=1,2,\ldots$, where $0=c_0<c_1<c_2<\cdots$. 
Note that $\bA_1=\{\bvartheta:(\bvartheta-\bmu)^T\bSigma^{-1}(\bvartheta-\bmu)\leq c_1\}$, and for $i\geq 2$, 
$\bA_i=\{\bvartheta:(\bvartheta-\bmu)^T\bSigma^{-1}(\bvartheta-\bmu)\leq c_i\}\setminus\cup_{j=1}^{i-1}\bA_j$.



To obtain reliable estimates of $\bmu$ and $\bSigma$, well-mixing MCMC algorithms may be employed. However, although a Gibbs sampling algorithm 
is available for our DP mixture model, it is difficult to get it converged in practice. To elucidate, note that as $\alpha\rightarrow\infty$,
a simple application of the Borel-Cantelli lemma in conjunction with the Markov inequality shows that $V_i$ converges to $0$ almost surely for each $i$. 
This entails that $w_j$, for $j=1,\ldots,N-1$, converge to zero, almost surely. Hence, conditional on the rest of the
unknowns, $\xi_1=\xi_2=\cdots=\xi_M$. That is, when all the $\xi_i$ are expected to be distinct, there is, in fact, only one common distinct value (or a small number
of distinct values) for these parameters in the relevant Gibbs sampling strategy, when $\alpha$ is large. Although theoretically the Gibbs sampler is still
irreducible, in practice, reliability of the chain in highly compromised, at least in our experience.

We completely avoid the aforementioned problem by implementing transformation based Markov Chain Monte Carlo (TMCMC) of \ctn{Dutta14} instead of the Gibbs sampler. 
In fact, additive TMCMC turned out to be adequate for all the examples that we considered. In our TMCMC algorithm we updated $\bvartheta$ using TMCMC and $\bxi$
by direct simulation from $G_N$. We set $\bmu$ and $\bSigma$ to be the mean and covariance of the TMCMC realizations of $\bvartheta$.


\subsection{Estimation of $\pi(\bB_i\times\bA_i|\by)$}
\label{subsec:mixing_probs}
Recall that the key idea of $iid$ sampling from $\pi(\btheta|\by)$ is to randomly select $\pi_i(\cdot|\by)$ with probability $\pi(\bB_i\times\bA_i|\by)$ 
and then to exactly simulate from $\pi_i(\cdot|\by)$. However, the mixing probabilities $\pi(\bB_i\times\bA_i|\by)$ are not available to us. 
Note that the TMCMC realizations are not useful for estimating these probabilities, since there can be only a finite number of such realizations in practice, 
whereas the number of the mixing probabilities is infinite. 
In this regard, we extend the Monte Carlo based estimation idea of \ctn{Bhatta21a} to suit our purpose,
assuming for the while that an infinite number of parallel processors are available, and that 
the $i$-th processor is used to estimate $\pi(\bB_i\times\bA_i|\by)$ using Monte Carlo sampling up to a constant.

To elaborate, let $\pi(\btheta|\by)=C\tilde\pi(\btheta|\by)$, where $\tilde\pi(\btheta|\by)$ is the right hand side of (\ref{eq:dp7}) and 
$C>0$ is the unknown normalizing constant. Then for any Borel set $\bA$ in the Borel $\sigma$-field
of $\mathbb R^d$, letting $\mathcal L(\bA)$ denote the
Lebesgue measure of $\bA$, observe that
\begin{align}
	\pi(\bB_i\times\bA_i|\by)&=C\mathcal L(\bA_i)\int \frac{\tilde\pi(\btheta|\by)}{\prod_{j=1}^MG_N(\xi_j)}
	\frac{1}{\mathcal L(\bA_i)}I_{\bA_i}(\bvartheta)I_{\bB_i}(\bxi)d\bvartheta\prod_{j=1}^MG_N(d\xi_j)\\
	&= C\mathcal L(\bA_i)E\left[\frac{\tilde\pi(\btheta|\by)}{\prod_{j=1}^MG_N(\xi_j)}I_{\bA_i}(\bvartheta)I_{\bB_i}(\bxi)\right],
	\label{eq:mc1}
\end{align}
the right hand side being $C\mathcal L(\bA_i)$ times the expectation of $\frac{\tilde\pi(\btheta|\by)}{\prod_{j=1}^MG_N(\xi_j)}I_{\bA_i}(\bvartheta)I_{\bB_i}(\bxi)$ 
with respect to the uniform distribution of $\bvartheta$ on $\bA_i$ and the distribution $\prod_{j=1}^MG_N(\xi_j)$ of $\bxi$ conditional on $\bvartheta$. This expectation
can be estimated by generating realizations of $\bvartheta$ from the uniform distribution on $\bA_i$, then drawing $\bxi$ from $\prod_{i=1}^MG_N(\xi_i)$
given $\bvartheta$ and subsequently evaluating $\frac{\tilde\pi(\btheta|\by)}{\prod_{j=1}^MG_N(\xi_j)}I_{\bA_i}(\bvartheta)I_{\bB_i}(\bxi)$ for the realizations and taking their
average. 
For the procedure of uniform sample generation from $\bA_i$ and computation of $\mathcal L(\bA_i)$, see \ctn{Bhatta21a}.

\subsection{Minorization for $\pi_i(\cdot|\by)$}
\label{subsec:minorization}

As in \ctn{Bhatta21a}, here we consider the following uniform independence proposal distribution on $\bA_i$ embedded in a 
Metropolis-Hastings framework for $\pi_i(\cdot|\by)$ to update the entire block $\bvartheta$:
\begin{equation}
	q_i(\bvartheta)=\frac{1}{\mathcal L(\bA_i)}I_{\bA_i}(\bvartheta).
	\label{eq:proposal}
\end{equation}
For further details regarding the usefulness of this proposal, see \ctn{Bhatta21a}.

For $\btheta\in\bB_i\times\bA_i$, for any Borel set $\mathbb B\times\mathbb A$ in the Borel $\sigma$-field of $\mathbb R^r$, where $r=d+M$, 
let $P_i(\btheta,\mathbb B\cap\bB_i\times\mathbb A\cap\bA_i)$ denote the corresponding Metropolis-Hastings transition probability for $\pi_i(\cdot|\by)$. 
Note that this transition probability is strictly positive only for those $\mathbb B$ such that $\mathbb B\cap\bB_i$ corresponds to
$\mathbb A\cap\bA_i$.
Let $s_i=\underset{\btheta\in\bB_i\times\bA_i}{\inf}~\frac{\tilde\pi(\btheta|\by)}{\prod_{j=1}^MG_N(\xi_j)}$ and 
$S_i=\underset{\btheta\in\bB_i\times\bA_i}{\sup}~\frac{\tilde\pi(\btheta|\by)}{\prod_{j=1}^MG_N(\xi_j)}$. 
Then, with (\ref{eq:proposal}) as the proposal density we have, for any $\btheta\in\bB_i\times\bA_i$: 
\begin{align}
	P_i(\btheta,\mathbb B\cap\bB_i\times\mathbb A\cap\bA_i)&\geq\int_{\mathbb B\cap\bB_i\times\mathbb A\cap\bA_i}
	\min\left\{1,\frac{\tilde\pi(\btheta'|\by)}{\tilde\pi(\btheta|\by)}\times\frac{\prod_{j=1}^MG_N(\xi_j)}{\prod_{j=1}^MG_N(\xi'_j)}\right\}q_i(\bvartheta')d\bvartheta'\prod_{j=1}^MG_N(d\xi'_j)\notag\\
	&\geq\left(\frac{s_i}{S_i}\right)\times\frac{\mathcal L(\mathbb A\cap\bA_i)}{\mathcal L(\bA_i)}\times \bG_N\left(\bxi'\in\mathbb B\cap\bB_i\right)\notag\\
	&=p_i~Q_i(\bvartheta'\in\mathbb A \cap\bA_i)\times\bG_N\left(\bxi'\in\mathbb B\cap\bB_i\right),
	\label{eq:minor1}
\end{align}
where $p_i=s_i/S_i$,
\begin{equation*}
	Q_i(\bvartheta'\in\mathbb A\cap\bA_i)=\frac{\mathcal L(\mathbb A\cap\bA_i)}{\mathcal L(\bA_i)}
\end{equation*}
is the uniform probability measure corresponding to (\ref{eq:proposal}), and
\begin{equation*}
	\bG_N\left(\bxi'\in\mathbb B\cap\bB_i\right)=\int_{\mathbb B\cap\bB_i}\prod_{j=1}^MG_N(d\xi'_j).	
\end{equation*}
Since (\ref{eq:minor1}) holds for all $\btheta\in\bB_i\times\bA_i$, the entire set $\bB_i\times\bA_i$ is a small set.


Let $\hat s_i$ and $\hat S_i$ denote the minimum and maximum of
$\frac{\tilde\pi(\btheta|\by)}{\prod_{j=1}^MG_N(\xi_j)}$
over the Monte Carlo samples drawn uniformly from $\bB_i\times\bA_i$ in course of estimating $\tilde\pi(\bB_i\times\bA_i|\by)$. 
Then $\frac{s_i}{S_i}\leq\frac{\hat s_i}{\hat S_i}$. Hence, there exists $\eta_i>0$ such that $1\geq\frac{s_i}{S_i}\geq\frac{\hat s_i}{\hat S_i}-\eta_i>0$.
Let $\hat p_i=\frac{\hat s_i}{\hat S_i}-\eta_i$. Then it follows from (\ref{eq:minor1}) that
\begin{equation}
	P_i(\btheta,\mathbb B\cap\bB_i\times\mathbb A\cap\bA_i)\geq \hat p_i~Q_i(\bvartheta'\in\mathbb A \cap\bA_i)\times \bG_N\left(\bxi'\in\mathbb B\cap\bB_i\right),
	\label{eq:minor1_hat}
\end{equation}
which we shall consider for our purpose.
Recall that in practice, $\eta_i$ is expected to be very close to zero, since the Monte Carlo sample size would be sufficiently large. 
Thus, $\hat p_i$ is expected to be very close to $p_i$.


\subsection{Split chain}
\label{subsec:coupling}
Due to the minorization (\ref{eq:minor1_hat}), the following decomposition holds for all $\btheta\in\bB_i\times\bA_i$:
\begin{align}
	P_i(\btheta,\mathbb B\cap\bB_i\times\mathbb A\cap\bA_i)&=\hat p_i~Q_i(\bvartheta'\in\mathbb A \cap\bA_i)\times \bG_N\left(\bxi'\in\mathbb B\cap\bB_i\right)\notag\\
	&\quad+(1-\hat p_i)~R_i(\btheta,\mathbb B\cap\bB_i\times\mathbb A\cap\bA_i),
	\label{eq:split1}
\end{align}
where
\begin{equation}
	R_i(\btheta,\mathbb B\cap\bB_i\times\mathbb A\cap\bA_i)=\frac{P_i(\btheta,\mathbb B\cap\bB_i\times\mathbb A\cap\bA_i)
	-\hat p_i~ Q_i(\bvartheta'\in\mathbb A\cap\bA_i)\times \bG_N\left(\bxi'\in\mathbb B\cap\bB_i\right)}{1-\hat p_i}
	\label{eq:split2}
\end{equation}
is the residual distribution.

Therefore, to implement the Markov chain $P_i(\btheta,\mathbb B\cap\bB_i\times\mathbb A\cap\bA_i)$, rather than proceeding directly with the uniform proposal based 
Metropolis-Hastings algorithm, we can
use the split (\ref{eq:split1}) to generate realizations from $P_i(\btheta,\mathbb B\cap\bB_i\times\mathbb A\cap\bA_i)$. That is, given $\btheta$, we can 
simulate from $Q_i\times\bG_N$ with probability $\hat p_i$, and with the remaining
probability, can generate from $R_i(\btheta,\cdot)$. 

To simulate from the residual density $R_i(\btheta,\cdot)$ we devise the following rejection sampling scheme. 
Let $\tilde R_i(\btheta,\btheta')$ and $\tilde P_i(\btheta,\btheta')$ 
denote the densities of $\btheta'$ corresponding to
$R_i(\btheta,\cdot)$ and $P_i(\btheta,\cdot)$, 
respectively. Then it follows from (\ref{eq:split1}) and (\ref{eq:split2}) that for all $\btheta\in\bB_i\times\bA_i$,
\begin{align*}
	&\tilde R_i(\btheta,\btheta')=\frac{\tilde P_i(\btheta,\btheta')-\hat p_i~ q_i(\bvartheta')\prod_{j=1}^MG_N(\xi'_j)}{1-\hat p_i}\leq \frac{\tilde P_i(\btheta,\btheta')}{1-\hat p_i}\notag\\
	&\Leftrightarrow \frac{\tilde R_i(\btheta,\btheta')}{\tilde P_i(\btheta,\btheta')}\leq \frac{1}{1-\hat p_i},~\mbox{for all}~\btheta'\in\bB_i\times\bA_i.
\end{align*}
Hence, given $\btheta$ we can continue to simulate $\btheta'\sim \tilde P_i(\btheta,\cdot)$ using the uniform proposal distribution (\ref{eq:proposal}) 
and generate $U\sim U(0,1)$ until 
\begin{equation}
U<\frac{(1-\hat p_i)\tilde R_i(\btheta,\btheta')}{\tilde P_i(\btheta,\btheta')}
	\label{eq:rejection_sampling2}
\end{equation}
is satisfied, at which point we accept $\btheta'$ as a realization from $\tilde R_i(\btheta,\cdot)$.

Now 
\begin{align}
	\tilde P_i(\btheta,\btheta')&=q_i(\bvartheta')\prod_{j=1}^MG_N(\xi'_j)
	\times\min\left\{1,\frac{\tilde\pi(\btheta'|\by)}{\tilde\pi(\btheta|\by)}\times\frac{\prod_{j=1}^MG_N(\xi_j)}{\prod_{j=1}^MG_N(\xi'_j)}\right\}
	+r_i(\btheta)I_{\btheta}(\btheta')\notag\\
	&=\frac{1}{\mathcal L(\bA_i)}\prod_{j=1}^MG_N(\xi'_j)\times\min\left\{1,\frac{\tilde\pi(\btheta'|\by)}{\tilde\pi(\btheta|\by)}
	\times\frac{\prod_{j=1}^MG_N(\xi_j)}{\prod_{j=1}^MG_N(\xi'_j)}\right\}+r_i(\btheta)I_{\btheta}(\btheta'),\notag
\end{align}
where
\begin{align}
	r_i(\btheta)&=1-\int_{\bB_i\times\bA_i}\min\left\{1,\frac{\tilde\pi(\btheta'|\by)}{\tilde\pi(\btheta|\by)}
	\times\frac{\prod_{j=1}^MG_N(\xi_j)}{\prod_{j=1}^MG_N(\xi'_j)}\right\}
	q_i(\bvartheta')d\bvartheta'\prod_{j=1}^MG_N(d\xi'_j)\notag\\
	&= 1-\int_{\bB_i\times\bA_i}\min\left\{1,\frac{\tilde\pi(\btheta'|\by)}{\tilde\pi(\btheta|\by)}
	\times\frac{\prod_{j=1}^MG_N(\xi_j)}{\prod_{j=1}^MG_N(\xi'_j)}\right\}\frac{1}{\mathcal L(\bA_i)}d\bvartheta'\prod_{j=1}^MG_N(d\xi'_j).
	\label{eq:mc_kernel2}
\end{align}
Let $\hat r_i(\btheta)$ denote the Monte Carlo estimate of $r_i(\btheta)$ obtained by simulating $\bvartheta'$ from the uniform distribution on $\bA_i$,
$\bxi'$ from $\prod_{j=1}^MG_N$ given $\bvartheta'$ and finally
taking the average of 
$\min\left\{1,\frac{\tilde\pi(\btheta'|\by)}{\tilde\pi(\btheta|\by)}\times\frac{\prod_{j=1}^MG_N(\xi_j)}{\prod_{j=1}^MG_N(\xi'_j)}\right\}$ 
in (\ref{eq:mc_kernel2}).
In our implementation, we shall consider the following:
\begin{align}
	\hat P_i(\btheta,\btheta')&=\frac{1}{\mathcal L(\bA_i)}\prod_{j=1}^MG_N(\xi'_j)\times\min\left\{1,\frac{\tilde\pi(\btheta'|\by)}{\tilde\pi(\btheta|\by)}
	\times\frac{\prod_{j=1}^MG_N(\xi_j)}{\prod_{j=1}^MG_N(\xi'_j)}\right\}+\hat r_i(\btheta)I_{\btheta}(\btheta'),~\mbox{and}\notag\\
	\hat R_i(\btheta,\btheta')&=\frac{\hat P_i(\btheta,\btheta')-\hat p_i~ q_i(\bvartheta')\prod_{j=1}^MG_N(\xi'_j)}{1-\hat p_i}.\notag
\end{align}
In all practical implementations, for sufficiently large Monte Carlo sample size, (\ref{eq:rejection_sampling2}) holds if and only if 
\begin{equation}
U<\frac{(1-\hat p_i)\hat R_i(\btheta,\btheta')}{\hat P_i(\btheta,\btheta')}
	\label{eq:rejection_sampling3}
\end{equation}
holds; see \ctn{Bhatta21a} for details. Consequently, as in \ctn{Bhatta21a}, 
we shall carry out our implementations with (\ref{eq:rejection_sampling3}).

\subsection{Perfect sampling from $\pi_i(\cdot|\by)$}
\label{subsec:perfect}
From (\ref{eq:split1}) it follows that (see \ctn{Bhatta21a}) at any given positive time $T_i=t$, $\btheta'$ will be drawn
from $Q_i\times\bG_N$ with probability $\hat p_i$. 
Hence, the distribution of $T_i$ is geometric, having the form
\begin{equation}
P(T_i=t)=\hat p_i (1-\hat p_i)^{t-1};~t=1,2,\ldots.
	\label{eq:geo}
\end{equation}

Due to (\ref{eq:geo}), first $T_i$ can be drawn from the geometric distribution and then one may simulate $\btheta^{(-T_i)}\sim Q_i\times\bG_N$. 
Subsequently, the chain only
needs to be carried forward in time till $t=0$, using $\btheta^{(t+1)}=\varrho_i(\btheta^{(t)},\bU^{(t+1)}_i)$,
where $\varrho_i(\btheta^{(t)},\bU^{(t+1)}_i)$ is the deterministic function corresponding to the simulation
of $\btheta^{(t+1)}$ from $\tilde R_i(\btheta^{(t)},\cdot)$.
Here $\{\bU^{(t)}_i;t=0,-1,-2,\ldots\}$ is an appropriate sequence of random numbers assumed to be available before beginning the perfect
sampling implementation. The realization $\btheta^{(0)}$ obtained at time $t=0$ is a perfect draw
from $\pi_i$. 

In practice, storing the uniform random numbers $\{\bU^{(t)}_i;t=0,-1,-2,\ldots\}$ or explicitly considering the deterministic relationship 
$\btheta^{(t+1)}=\varrho_i(\btheta^{(t)},\bU^{(t+1)}_i)$,
are not required. These would be required only if we had taken the search approach, namely, iteratively starting the Markov chain at all initial values at negative times
and carrying the sample paths to zero.

The complete algorithm for $iid$ sample generation from $\pi(\cdot|\by)$ is of the same form as Algorithm 1 of \ctn{Bhatta21a}, and hence we do not provide
the explicit algorithm here.

\subsection{Diffeomorphism}
\label{subsec:diffeo}

It is obvious that small values of $\hat p_i$ would lead to large values of $T_i$, which would make the perfect sampling algorithm inefficient. 
To solve this problem, \ctn{Bhatta21a} proposed inversion of a diffeomorphism proposed
in \ctn{Johnson12a} to flatten
the posterior distribution in a way that its infimum and the supremum are reasonably close (so that $\hat p_i$ are adequately large) on all the relevant
ellipsoids and annuli. 

The issue of small values of $\hat p_i$ persists even the current DP mixture context, and hence the diffeomorphism fix is again of great value. 
Here it is of interest to render the posterior $\pi(\bxi,\bvartheta|\by)$ thick-tailed using the inverse of the diffeomorphic transformation of \ctn{Johnson12a}. 
Note, however, that since  $\bxi$ depends directly on $\bvartheta$ through $G_N$, it is sufficient to
consider the inverse diffeomorphic transformation for $\bvartheta$ only. 

Thus, setting $\bgamma=h^{-1}(\bvartheta)$, where $h$ is a diffeomorphism, the density of
$(\bxi,\bgamma)$ is given by
\begin{align}
	\pi_{\bxi,\bgamma}(\bxi,\bgamma|\by)=\pi\left(\bxi,h^{-1}(\bgamma)|\by\right)\left|\mbox{det}~\nabla h(\bgamma)\right|^{-1}
\label{eq:transformed_target}
\end{align}
where $\nabla h(\bgamma)$ denotes the gradient of $h$ at $\bgamma$ and $\mbox{det}~\nabla h(\bgamma)$ 
stands for the determinant of the gradient of $h$ at $\bgamma$. 
The details of the transformation are provided below.

As in \ctn{Bhatta21a}, here we consider the  
following isotropic function $h:\mathbb R^d\mapsto\mathbb R^d$ of \ctn{Johnson12a}:
\begin{equation}
	h(\bgamma)=\left\{\begin{array}{cc}f(\|\bgamma\|)\frac{\bgamma}{\|\bgamma\|}, & \bgamma\neq \bzero\\
		0, & \bgamma=\bzero,
\end{array}\right.
\label{eq:isotropy}
\end{equation}
for some function $f: (0,\infty)\mapsto (0,\infty)$, $\|\cdot\|$ being the Euclidean norm.
\ctn{Johnson12a} consider isotropic diffeomorphisms, that is, functions of the form $h$ where 
both $h$ and $h^{-1}$ are continuously differentiable, with the further property that 
$\mbox{det}~\nabla h$ and  $\mbox{det}~\nabla h^{-1}$ are also continuously 
differentiable. Specifically, they define 
$f:[0,\infty)\mapsto [0,\infty)$ given by
\begin{equation}
f(x)=\left\{\begin{array}{cc}e^{bx}-\frac{e}{3}, & x>\frac{1}{b}\\
x^3\frac{b^3e}{6}+x\frac{be}{2}, & x\leq \frac{1}{b},
\end{array}\right.
\label{eq:diffeo2}
\end{equation}
where $b>0$. 

We apply the same transformation to the uniform proposal density (\ref{eq:proposal}), so that the new proposal density becomes
\begin{equation}
	q_i(\bgamma)=\frac{1}{\mathcal L(\bA_i)}I_{\bA_i}(h^{-1}(\bgamma))\left|\mbox{det}~\nabla h(\bgamma)\right|^{-1}.
	\label{eq:proposal2}
\end{equation}

Now, for any set $\bA$, let $h_{\bvartheta}(\bA)=\left\{h(\bvartheta):\bvartheta\in\bA\right\}$ and for any set $\bB$, let
$h_{\bxi}(\bB)=\left\{h(\bxi):\bxi\in\bB\right\}$.
Also, let $s_i=\underset{\bxi\in h(\bB_i),\bgamma\in h(\bA_i)}{\inf}~\frac{\tilde\pi_{\bxi,\bgamma}(\bxi,\bgamma|\by)}{q_i(\bgamma)\prod_{j=1}^MG_N(\xi_j)}$ and 
$S_i=\underset{\bxi\in h(\bB_i),\bgamma\in h(\bA_i)}{\sup}~\frac{\tilde\pi_{\bxi,\bgamma}(\bxi,\bgamma|\by)}{q_i(\bgamma)\prod_{j=1}^MG_N(\xi_j)}$, 
where $\tilde\pi_{\bxi,\bgamma}(\bxi,\bgamma|\by)$ 
is the same as (\ref{eq:transformed_target})
but without the normalizing constant.
Then, with (\ref{eq:proposal2}) as the proposal density, we have 
\begin{align}
	&P_i((\bxi,\bgamma),h_{\bxi}(\mathbb B\cap\bB_i)\times h_{\bvartheta}(\mathbb A\cap\bA_i))\notag\\
	&\geq\int_{h_{\bxi}(\mathbb B\cap\bB_i)\times h_{\bvartheta}(\mathbb A\cap\bA_i)}
	\min\left\{1,\frac{\tilde\pi_{\bxi,\bgamma}(\bxi',\bgamma'|\by)}{\tilde\pi_{\bxi,\bgamma}(\bxi,\bgamma|\by)}
	\times\frac{q_i(\bvartheta)}{q_i(\bvartheta')}\times\frac{\prod_{j=1}^NG_N(\xi_j)}{\prod_{j=1}^NG_N(\xi'_j)}\right\}
	q_i(\bgamma')d\bgamma'\prod_{j=1}^MG_N(d\xi'_j)\notag\\
	&\geq p_i~Q_i\left(h_{\bvartheta}(\mathbb A \cap\bA_i)\right)\times\bG_N\left(h_{\bxi}(\mathbb B\cap\bB_i)\right),\notag
\end{align}
where $p_i=s_i/S_i$, $Q_i$ is the probability measure corresponding to (\ref{eq:proposal2}) and $\bG_N$ is given by
\begin{equation*}
	\bG_N\left(\bxi'\in h_{\bxi}(\mathbb B\cap\bB_i)\right)=\int_{h_{\bxi}(\mathbb B\cap\bB_i)}\prod_{j=1}^MG_N(d\xi'_j).	
\end{equation*}
With $\hat p_i=\hat s_i/\hat S_i-\eta_i$, where $\hat s_i$ and $\hat S_i$ are Monte Carlo estimates of $s_i$ and $S_i$ and $\eta_i>0$ is adequately small, 
the rest of the details remain the same as before with necessary modifications pertaining to the new proposal density (\ref{eq:proposal2}) and the new 
Metropolis-Hastings acceptance ratio
with respect to (\ref{eq:proposal2}) incorporated in the subsequent steps. Once $\bgamma$ is generated from (\ref{eq:transformed_target}) we
transform it back to $\bvartheta$ using $\bvartheta=h^{-1}(\bgamma)$ and accordingly reset the values of $\bxi$.

\section{Applications}
\label{sec:applications}
We now illustrate our $iid$ sampling idea on posterior DP mixture of normal mixture models with unknown but bounded number of components
with application to the well-studied enzyme, acidity and the galaxy data sets. \ctn{Richardson97} and \ctn{Das19}
modeled these data sets using parametric normal mixtures and applied reversible jump Markov chain Monte Carlo and transdimensional
transformation based Markov chain Monte Carlo, respectively, for Bayesian inference.

On the other hand, \ctn{Bhattacharya08} modeled these data using 
the DP mixture of the form given in Section \ref{sec:dp_bounded} with an $M$-component mixture of normal densities.
In other words, $f(\cdot|\xi_j)$ is taken as the density of $N(\nu_j,\sigma^2_j)$, the normal distribution with mean $\nu_j$ and variance $\sigma^2_j$, 
the latter primarily parameterized by $\lambda_j=\sigma^{-2}_j$.
Further, he set $\pi_j=1/M$, for $j=1,\ldots,M$; this choice may be advantageous in real data setups, as aptly demonstrated in \ctn{Majumdar13}.
Integrating out $G$, \ctn{Bhattacharya08} arrived at a P\'{o}lya-urn scheme, which he used to construct a Gibbs sampler for Bayesian inference. 

For our illustration, we consider the same model and priors as \ctn{Bhattacharya08} but implement the $iid$ sampling method for the three aforementioned datasets.
It is to be noted that the P\'{o}lya-urn based Gibbs sampling procedure will not serve our purpose of estimating $\bmu$ and $\bSigma$ needed for $\bA_i$,
as the random measure $G$ is integrated out. Indeed, recall from Section \ref{subsec:choice_sets} that $\bA_i$ are based upon $\bvartheta$, 
which includes parameters associated with $G$. In the same section we argued that Gibbs sampling including $G_N$ is laden with difficulties, and that such
difficulties can be completely bypassed using TMCMC, which we employ and generally recommend for estimating $\bmu$ and $\bSigma$.

All our $iid$ simulations are based on the diffeomorphic transformation detailed in Section \ref{subsec:diffeo}, since without this substantially large values
of $\hat p_i$ could not be ensured.

All our codes are written in C using the Message
Passing Interface (MPI) protocol for parallel processing. We implemented our codes on a
80-core VMWare provided by Indian Statistical Institute. The machine has 2 TB memory
and each core has about 2.8 GHz CPU speed.

Below we provide details on $iid$ sampling for the three datasets, along with comparisons with TMCMC.

\subsection{Enzyme data}
\label{subsection:enzyme}
This dataset concerns the distribution of enzymatic activity in the
blood, for an enzyme involved in the metabolism of carcinogenic substances,
among a group of $n=245$ unrelated individuals. We model this data using normal mixture of a maximum of $M=30$ components, where the parameters are assumed
to arise from $G_N$, with $N=50$. The choice of $M$ is the same as in \ctn{Bhattacharya08}, \ctn{Das19}, \ctn{Sabya12}, \ctn{Richardson97}, while that of $N$ 
is based upon Remark \ref{remark:truncation}. 

As in \ctn{Bhattacharya08}, we assume that under $G_0$, $\tau_j\sim\mathcal G(s/2,S/2)$ and given $\tau_j$, $\nu_j\sim N\left(\nu_0,\frac{c}{\tau_j}\right)$, 
where for $a>0$, $b>0$, $\mathcal G(a,b)$ stands for the gamma distribution with mean $a/b$ and
variance $a/b^2$, and $c>0$ is an appropriate constant. In this example, following \ctn{Bhattacharya08} we set $s=4$, $S=2\times(0.2/1.22)=0.328$, $\nu_0=1.45$, $c=33.3$
For the prior of $\alpha$ we considered $\mathcal G\left(a_{\alpha},b_{\alpha}\right)$ with $a_{\alpha}=2$ and $b_{\alpha}=4$, as in \ctn{Bhattacharya08}.

To estimate $\bmu$ and $\bSigma$ for $\bA_i$, we implemented additive TMCMC with scaling constants in the additive transformation for $\bvartheta$ set to $\sqrt{0.5}$,
while $\bxi$ are simulated from $\prod_{j=1}^MG_N$, given $\bvartheta$.
We discarded the first $10^6$ iterations as burn-in and stored one in $100$ iterations in the next $10^6$ iterations, to yield $10,000$ realizations for our purpose.
This exercise took $25$ minutes on a single core.

To complete specification of $\bA_i$, we set $\sqrt{c_1}=8.0$ and $\sqrt{c_i}=\sqrt{c_1}+0.0005\times (i-1)$, for $i=1,\ldots,10^4$. These choices ensured 
adequate coverage of the parameter space of $\bvartheta$ and significantly large values of $\hat p_i$ as we 
chose the diffeomorphism parameter $b=0.01$. We sampled $5000$ Monte Carlo realizations uniformly from $\bA_i$ to reliably estimate the corresponding probabilities
and to compute $\hat s_i$ and $\hat S_i$; we set $\eta_i=10^{-10}$. In the Monte Carlo context, we replaced the computationally inefficient 
rejection sampling method of uniformly sampling from $\bA_i$ with the efficient algorithm proposed in \ctn{Bhatta21a}, completely bypassing rejection sampling. 

With these, we simulated $10,000$ $iid$ realizations from the posterior $\pi(\bxi,\bvartheta|\by)$ on $80$ cores, which took $19$ minutes. Using these $iid$ realizations
we obtained the key results
presented diagrammatically in Figure \ref{fig:enzyme}.
Panel (a) of Figure \ref{fig:enzyme} compares the posterior predictive densities obtained using TMCMC and $iid$ sampling, showing that they are almost identical
and well-capture the details of the histogram of the observed data.  
Panel (b) compares $20$ times pointwise posterior predictive variances associated with panel (a) computed using TMCMC and $iid$ realizations. Although $iid$-based variances 
are expected to be non-negligibly larger than those based on TMCMC, here they are only
slightly larger than those of TMCMC, in spite of scaling up by $20$. The reason for such close agreement between $iid$ and TMCMC realizations is excellent
mixing of the TMCMC chain, as summarized by the typical autocorrelation plots of $\nu_{30}$ and $\tau_{30}$, provided in panels (c) and (d), respectively.  

Letting $K$ denote the number of mixture components, with respect to $iid$ sampling, the sample-based posterior probabilities of $K=2$, $3$, $4$, $5$, $6$ are
$0.2758$, $0.4462$, $0.2280$, $0.0461$, $0.0039$, respectively and zero for the other values of $K$. On the other hand, the TMCMC based posterior probabilities
of the same values of $K$ are $0.2743$, $0.4468$, $0.2313$, $0.0443$, $0.0033$ and zero for the other values of $K$. Thus, a strong agreement is exhibited between
$iid$ sampling and TMCMC even with respect to the posterior of $K$.
\begin{figure}
	\centering
	\subfigure [TMCMC and $iid$-based posterior predictive density for enzyme.]{ \label{fig:enzyme1}
	\includegraphics[width=7.5cm,height=7.5cm]{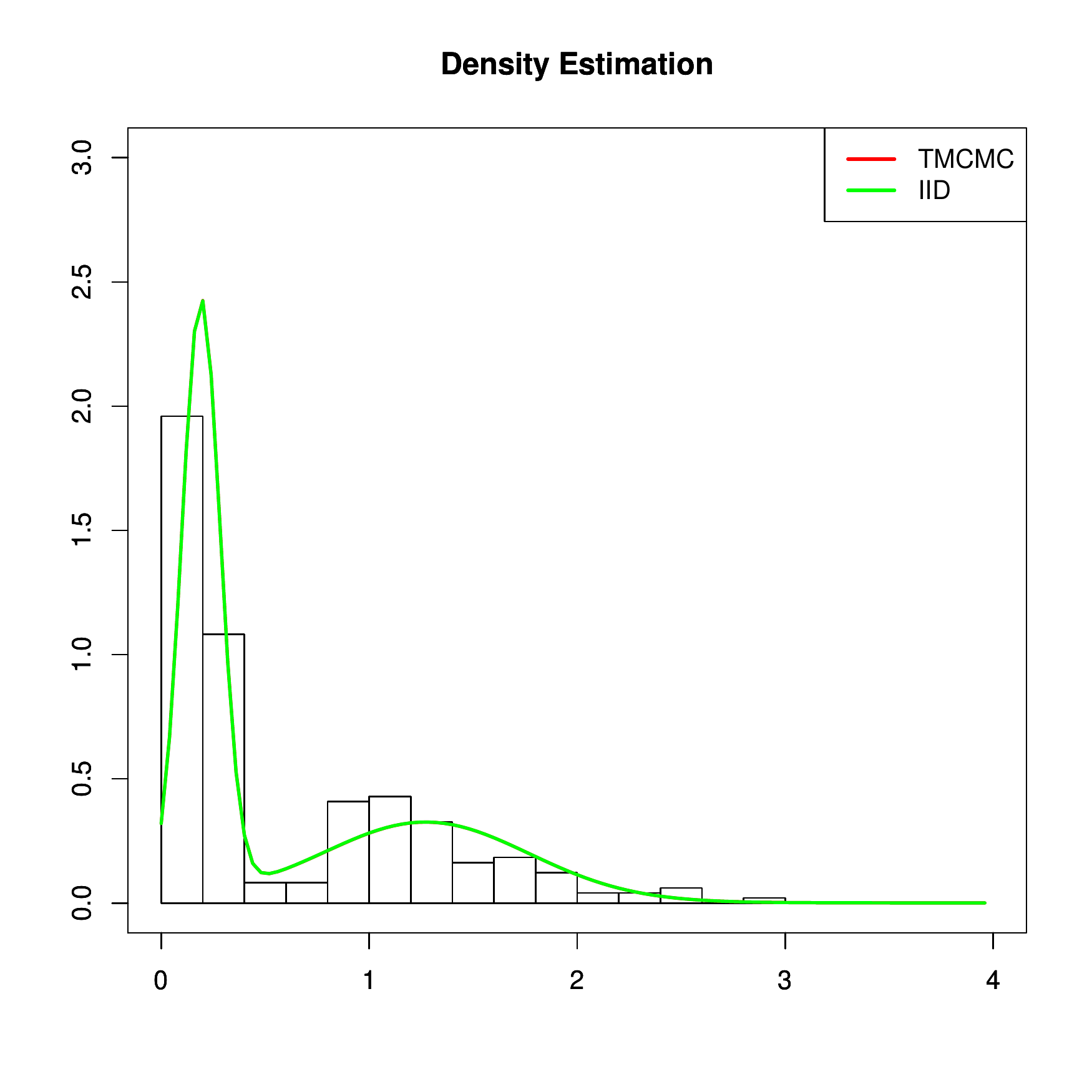}}
	\hspace{2mm}
	\subfigure [Poinwise variances with TMCMC and $iid$ sampling.]{ \label{fig:enzyme2}
	\includegraphics[width=7.5cm,height=7.5cm]{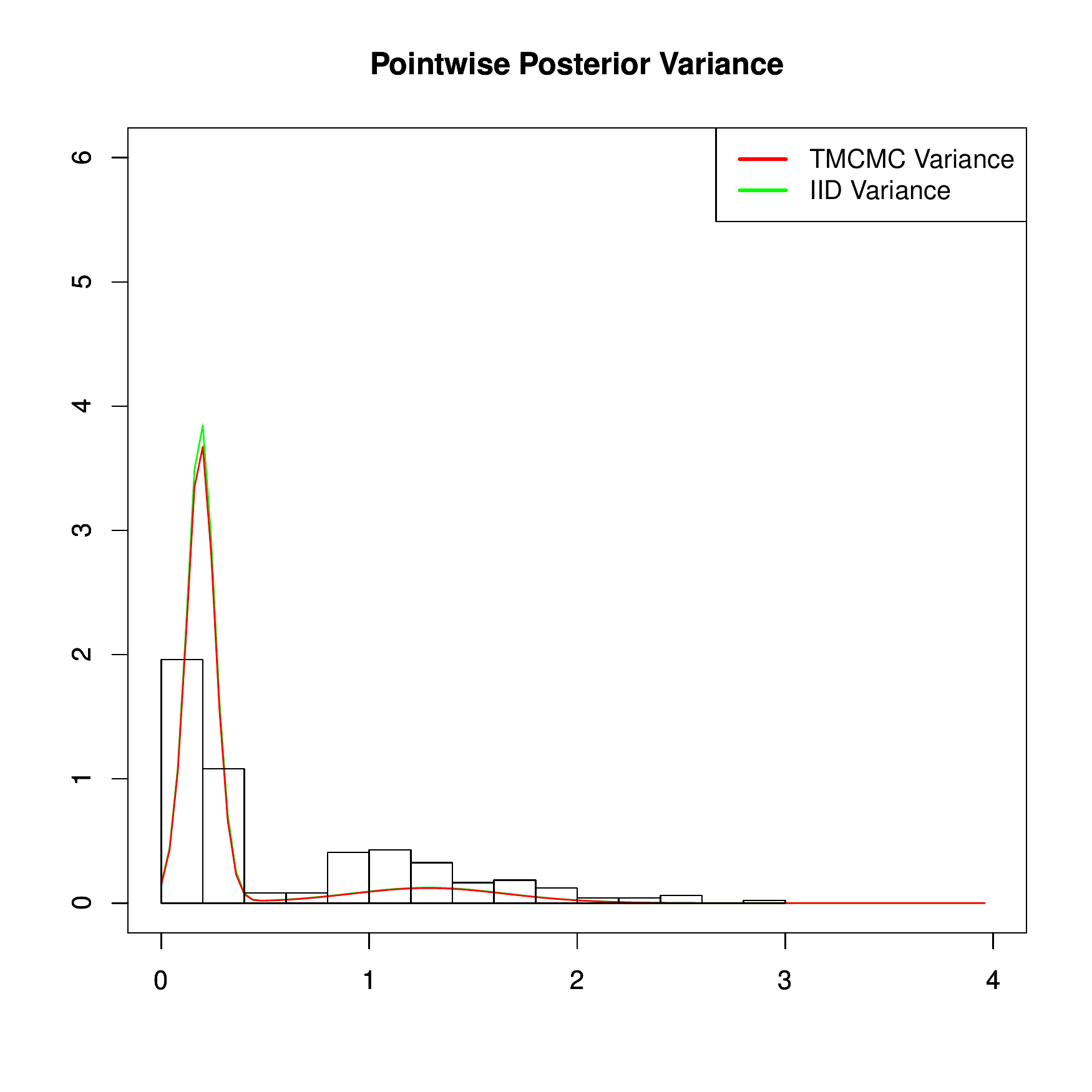}}\\
	\vspace{2mm}
	\subfigure [TMCMC autocorrelation plot for $\nu_{30}$.]{ \label{fig:enzyme3}
	\includegraphics[width=7.5cm,height=7.5cm]{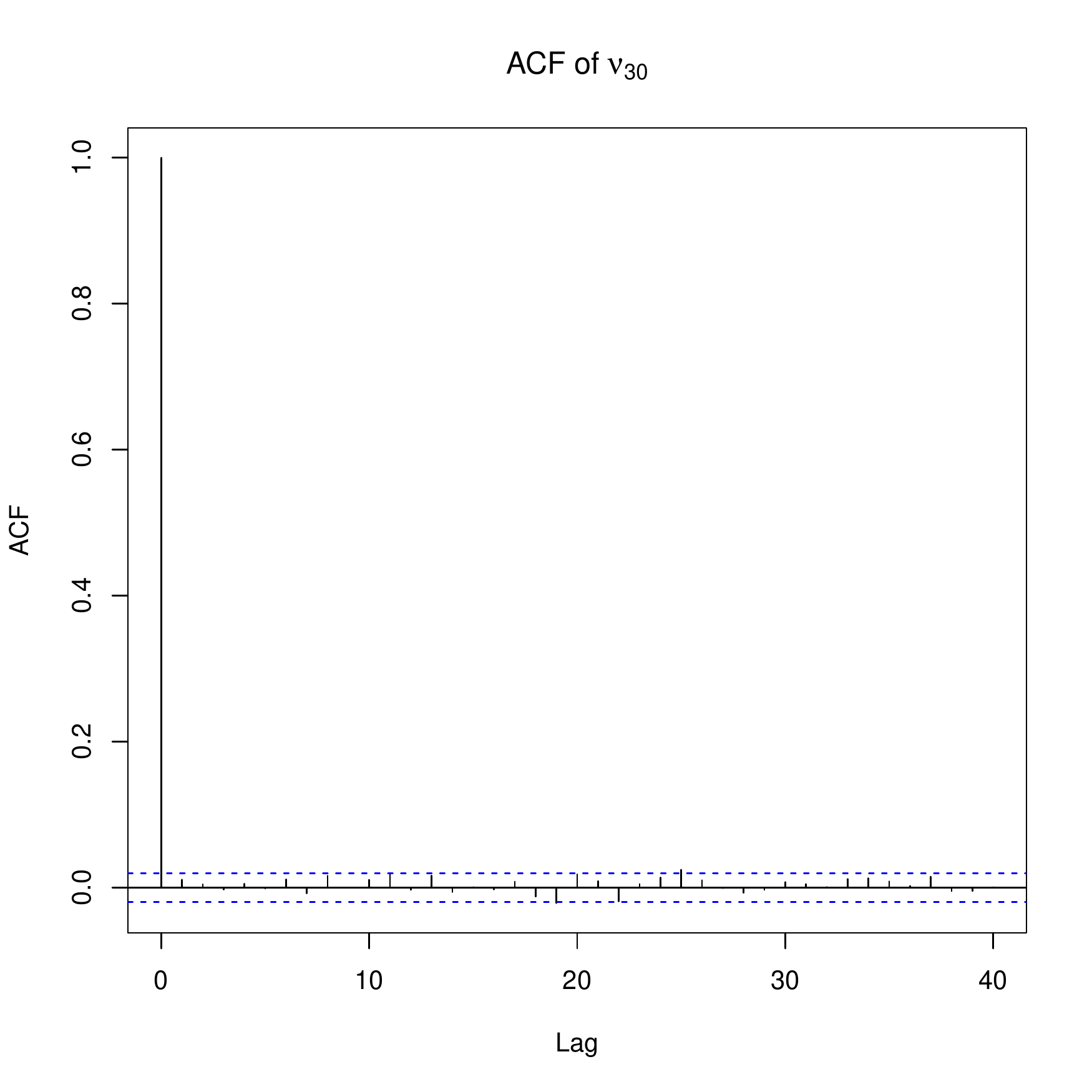}}
	\vspace{2mm}
	\subfigure [TMCMC autocorrelation plot for $\tau_{30}$.]{ \label{fig:enzyme4}
	\includegraphics[width=7.5cm,height=7.5cm]{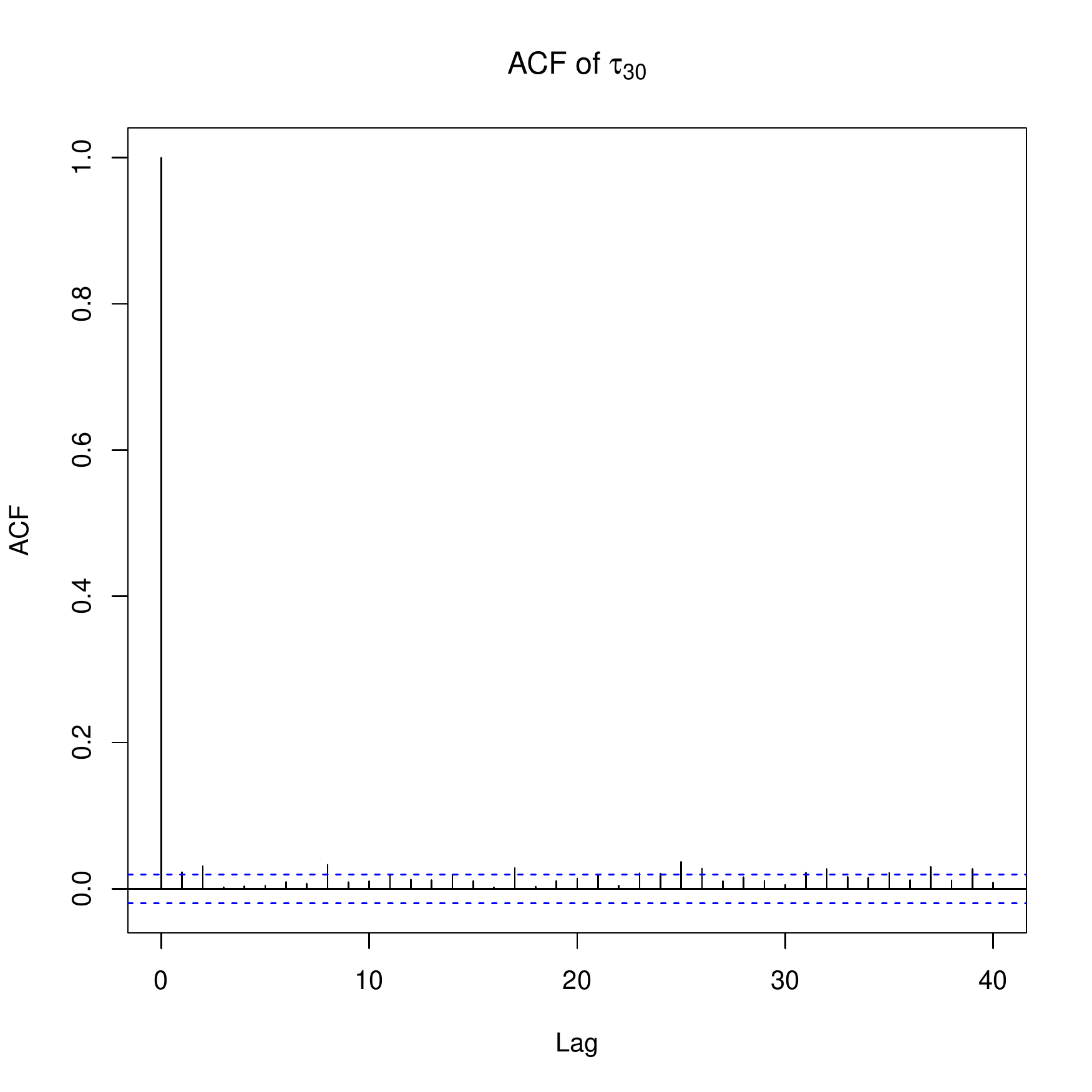}}
	\caption{IID and TMCMC sampling for the enzyme data.}
	\label{fig:enzyme}
\end{figure}

\subsection{Acidity data}
\label{subsection:acidity}
The acidity data set is on an acidity index measured in a sample of $n=155$ lakes in north-central Wisconsin.
With the same model and prior structure as for the enzyme data, here we set $s=4$, $S=2\times(0.2/0.573)=0.698$, $\nu_0=5.02$, $c=33.3$, $a_{\alpha}=2$, $b_{\alpha=4}$,
$M=30$ and $N=50$.

The TMCMC details remain essentially the same as in the enzyme case. Only here we discarded the first $1.5\times 10^6$ iterations as burn-in and stored one in $150$ iterations 
in the next $1.5\times 10^6$ iterations, to yield $10,000$ realizations. The scaling constants for additive transformation for $\bvartheta$ here are $\sqrt{0.05}$.
This exercise took $28$ minutes on a single core.

The $iid$ sampling details are also essentially the same as in the enzyme data; only here  
we set $\sqrt{c_1}=7.0$ and $\sqrt{c_i}=\sqrt{c_1}+0.0005\times (i-1)$, for $i=1,\ldots,10^4$.  
On $80$ cores, the time taken is only $8$ minutes to generate $10,000$ $iid$ realizations.

Figure \ref{fig:acidity} presents the results of $iid$ sampling for the acidity data, along with comparison with TMCMC. Panel (a) shows close agreement
between $iid$ and TMCMC sampling, but not as close as for enzyme. Indeed, panel (b) shows that the TMCMC based pointwise posterior predictive variances, multiplied by $38$, 
are uniformly non-negligibly smaller than those based on $iid$ realizations. The reason for this difference can be attributed to the TMCMC autocorrelations. Although
the location parameters have negligible autocorrelations, exemplified by $\nu_{30}$, shown in panel (c), the scale parameters $\tau$ do not have negligible autocorrelations
for many lags, as shown in panel (d) for $\tau_{30}$ as an instance.

Here the number of components $K=2$, $3$, $4$, $5$ has the empirical posterior probabilities $0.7810$, $0.2111$, $0.0078$, $0.0001$ and zero
for other values of $K$ with respect to $iid$ sampling and 
$0.7289$, $0.2590$, $0.0121$, $0.0000$ and zero for other values of $K$ with respect to TMCMC, which are not in disagreement.
\begin{figure}
	\centering
	\subfigure [TMCMC and $iid$-based posterior predictive density for acidity.]{ \label{fig:acidity1}
	\includegraphics[width=7.5cm,height=7.5cm]{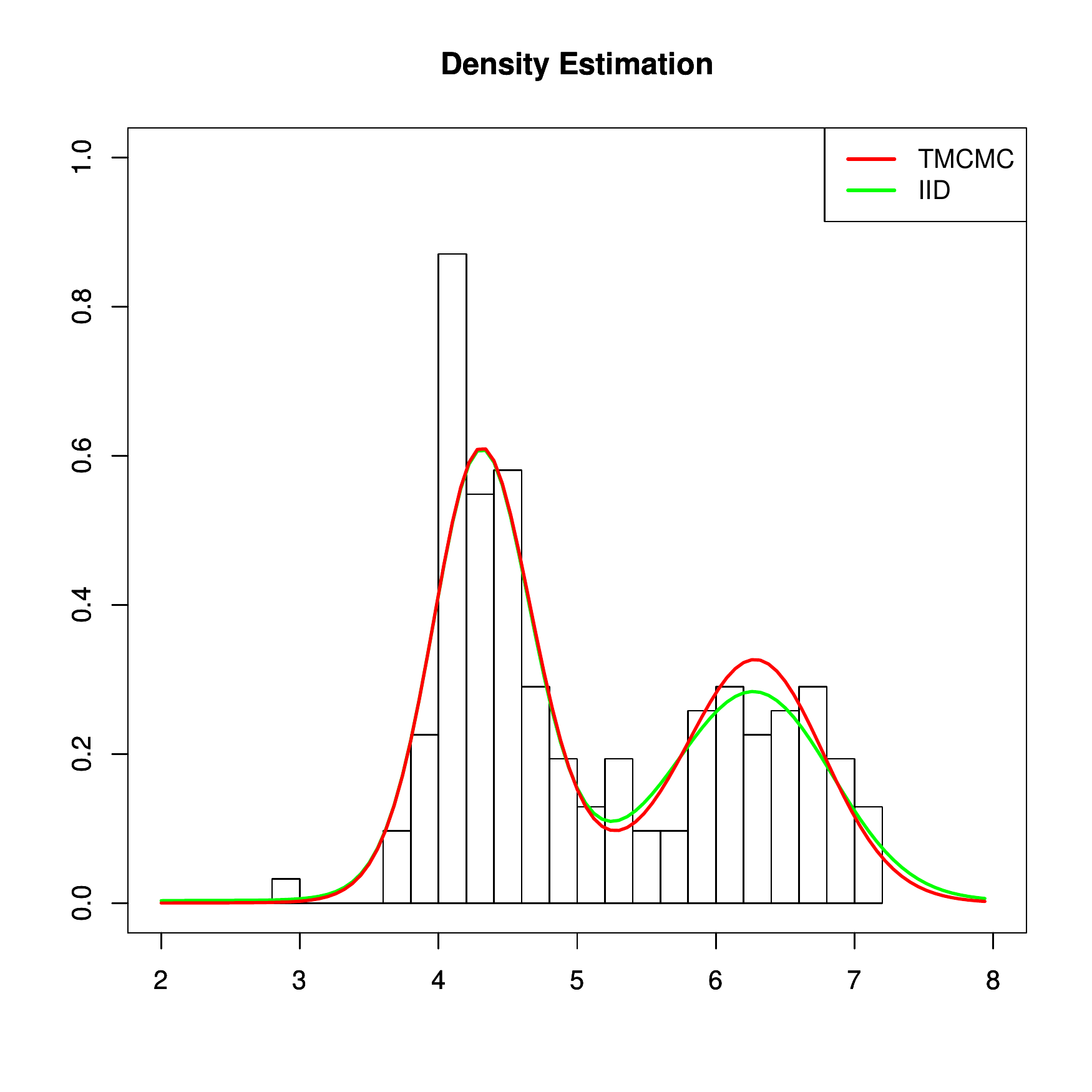}}
	\hspace{2mm}
	\subfigure [Poinwise variances with TMCMC and $iid$ sampling.]{ \label{fig:acidity2}
	\includegraphics[width=7.5cm,height=7.5cm]{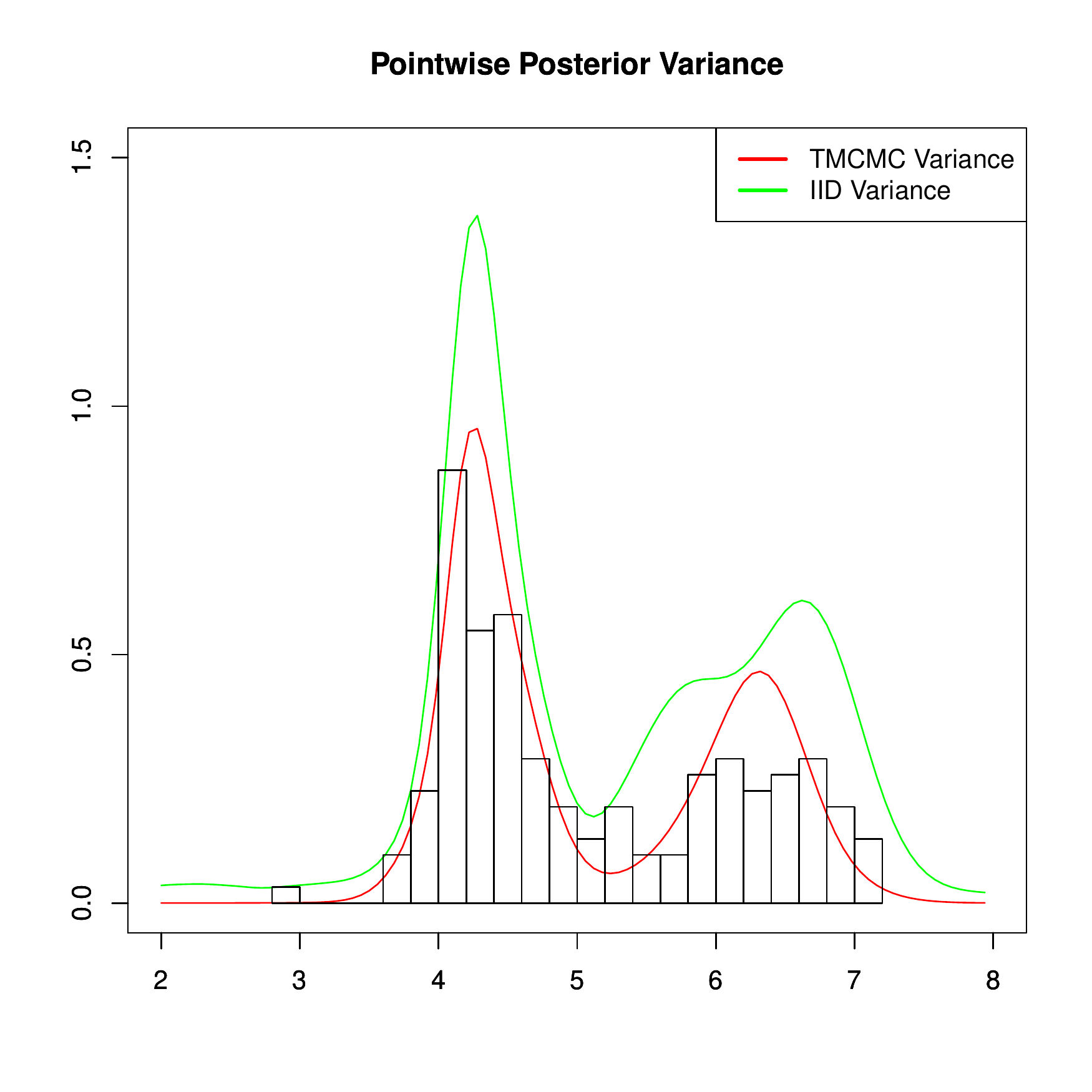}}\\
	\vspace{2mm}
	\subfigure [TMCMC autocorrelation plot for $\nu_{30}$.]{ \label{fig:acidity3}
	\includegraphics[width=7.5cm,height=7.5cm]{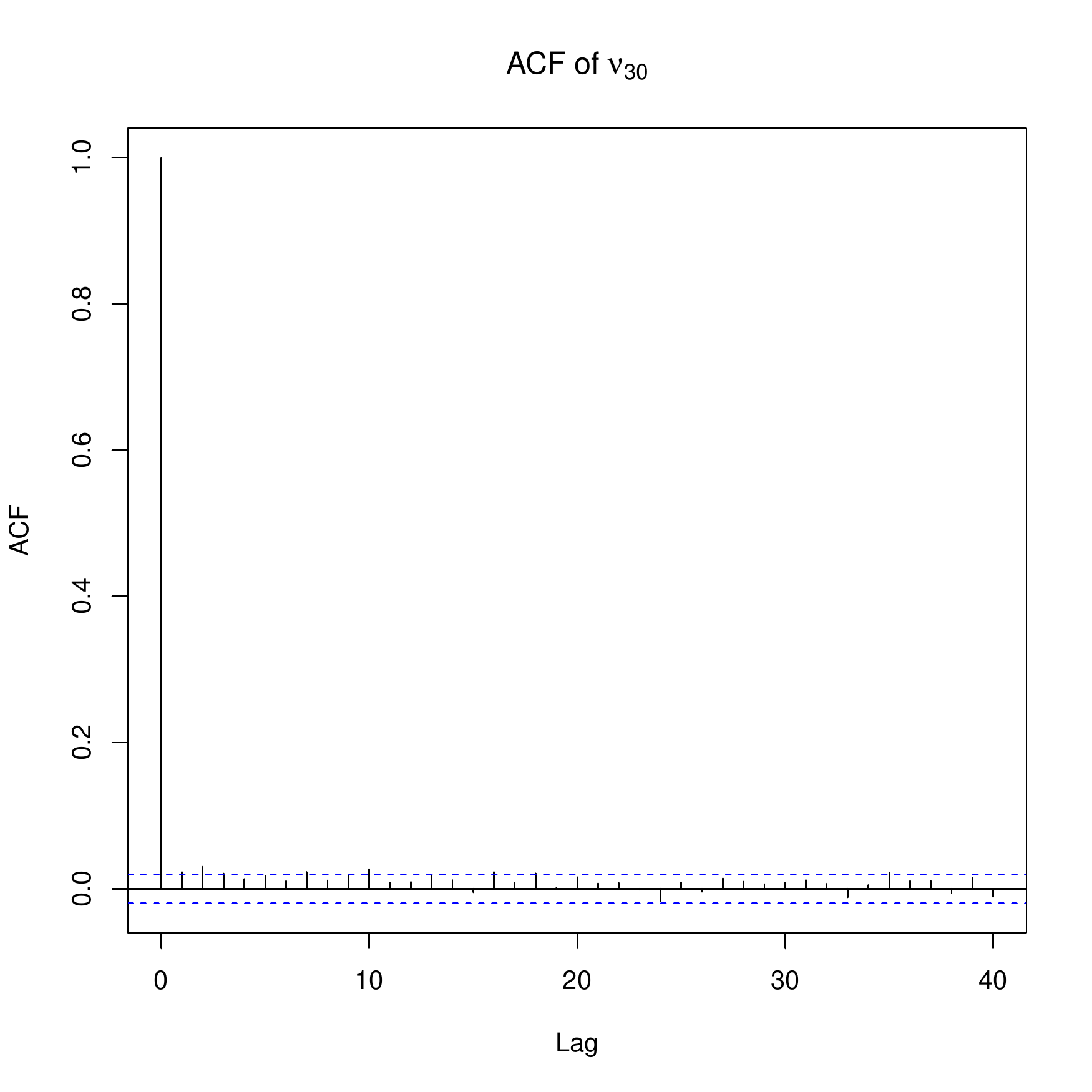}}
	\vspace{2mm}
	\subfigure [TMCMC autocorrelation plot for $\tau_{30}$.]{ \label{fig:acidity4}
	\includegraphics[width=7.5cm,height=7.5cm]{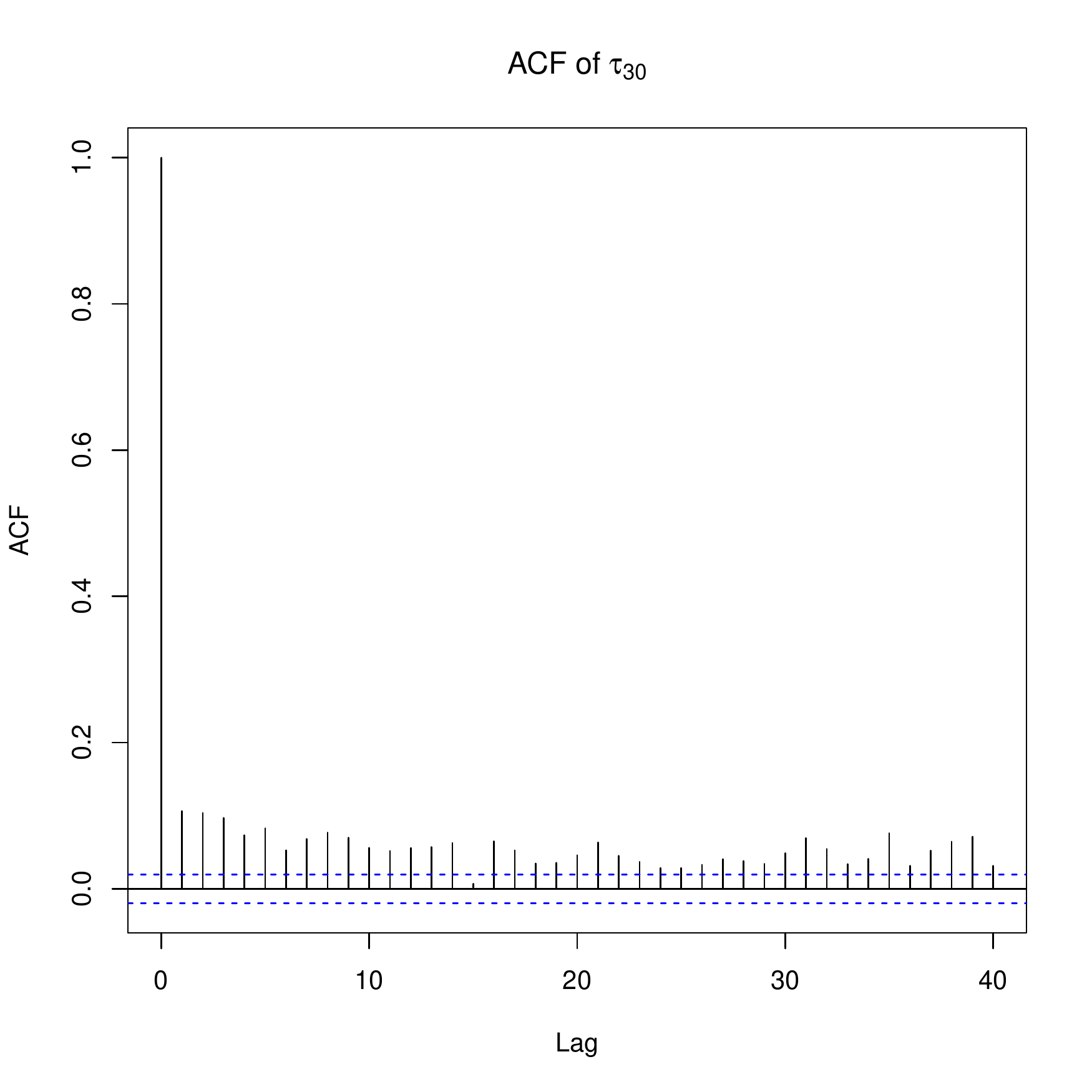}}
	\caption{IID and TMCMC sampling for the acidity data.}
	\label{fig:acidity}
\end{figure}

\subsection{Galaxy data}
\label{subsection:galaxy}
The galaxy data consists of the velocities of $n=82$ distant galaxies, diverging from our own galaxy.
With the same model and prior structure as before, here we set $s=4$, $S=2$, $\nu_0=20$, $c=33.3$, $a_{\alpha}=2$, $b_{\alpha=4}$,
$M=30$ and $N=50$.

The TMCMC details here are the same as in the acidity case, except that here we set the scaling constants for additive transformation of $\bvartheta$ to be $1$.
The time taken is $21$ minutes on a single core for this TMCMC algorithm for the galaxy data.

The $iid$ sampling details are essentially the same as the previous two examples, except that here 
$\sqrt{c_1}=9.0$ and $\sqrt{c_i}=\sqrt{c_1}+0.0005\times (i-1)$, for $i=1,\ldots,10^4$ and the diffeomorphism parameter is $b=0.001$.
The time taken for generating $10,000$ $iid$ realizations is only $5$ minutes on our $80$ cores.

Figure \ref{fig:galaxy} presents the results of $iid$ sampling and TMCMC for the galaxy data. Here again panel (a) shows close agreement
between $iid$ and TMCMC sampling; the only slight disagreement being at the left-most mode. 
Panel (b) shows that the TMCMC based pointwise posterior predictive variances, again multiplied by $38$, 
are non-negligibly smaller than those based on $iid$ realizations, except at a few points in the left-most modal region. 
The reason for this difference can be attributed to the TMCMC autocorrelations. Although
both location and scale parameters seem to have small autocorrelations, shown in panels (c) and (d), these are of course somewhat high in comparison with 
the $iid$ case where no autocorrelation is present, and have hence contributed to the slight disagreement in panel (a).

The posterior probabilities of the number of components $K=1$, $2$, $3$, $4$, $5$ are $0.0265$, $0.2725$, $0.4994$, $0.1965$, $0.0051$ and zero for other values
of $K$ with respect to the $iid$ sampling procedure and those with respect to TMCMC are $0.0229$, $0.2517$, $0.5185$, $0.2045$, $0.0024$ and zero for other values of $K$. 
That is, with respect to the number of components as well, the posterior probabilities are in agreement.
\begin{figure}
	\centering
	\subfigure [TMCMC and $iid$-based posterior predictive density for galaxy.]{ \label{fig:galaxy1}
	\includegraphics[width=7.5cm,height=7.5cm]{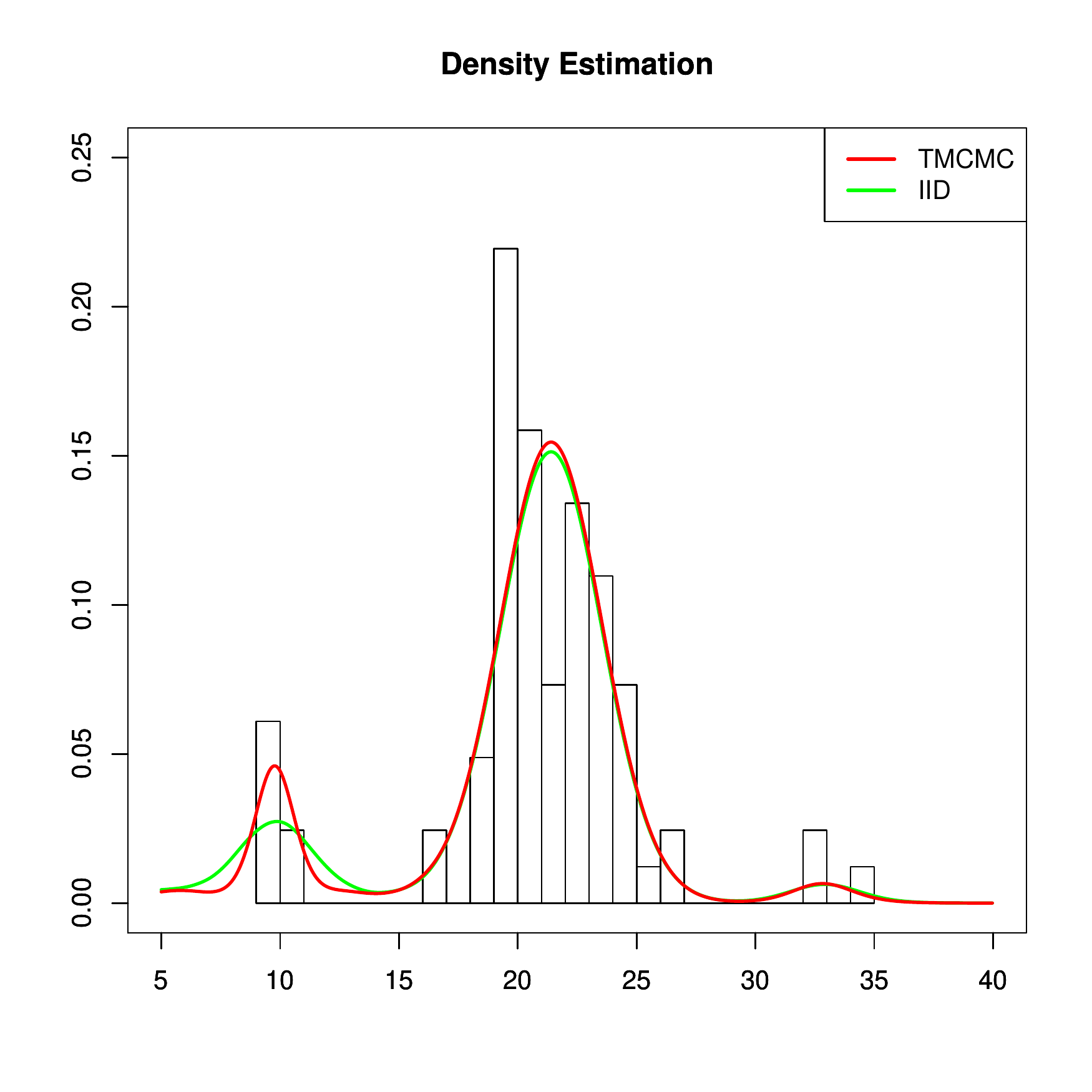}}
	\hspace{2mm}
	\subfigure [Poinwise variances with TMCMC and $iid$ sampling.]{ \label{fig:galaxy2}
	\includegraphics[width=7.5cm,height=7.5cm]{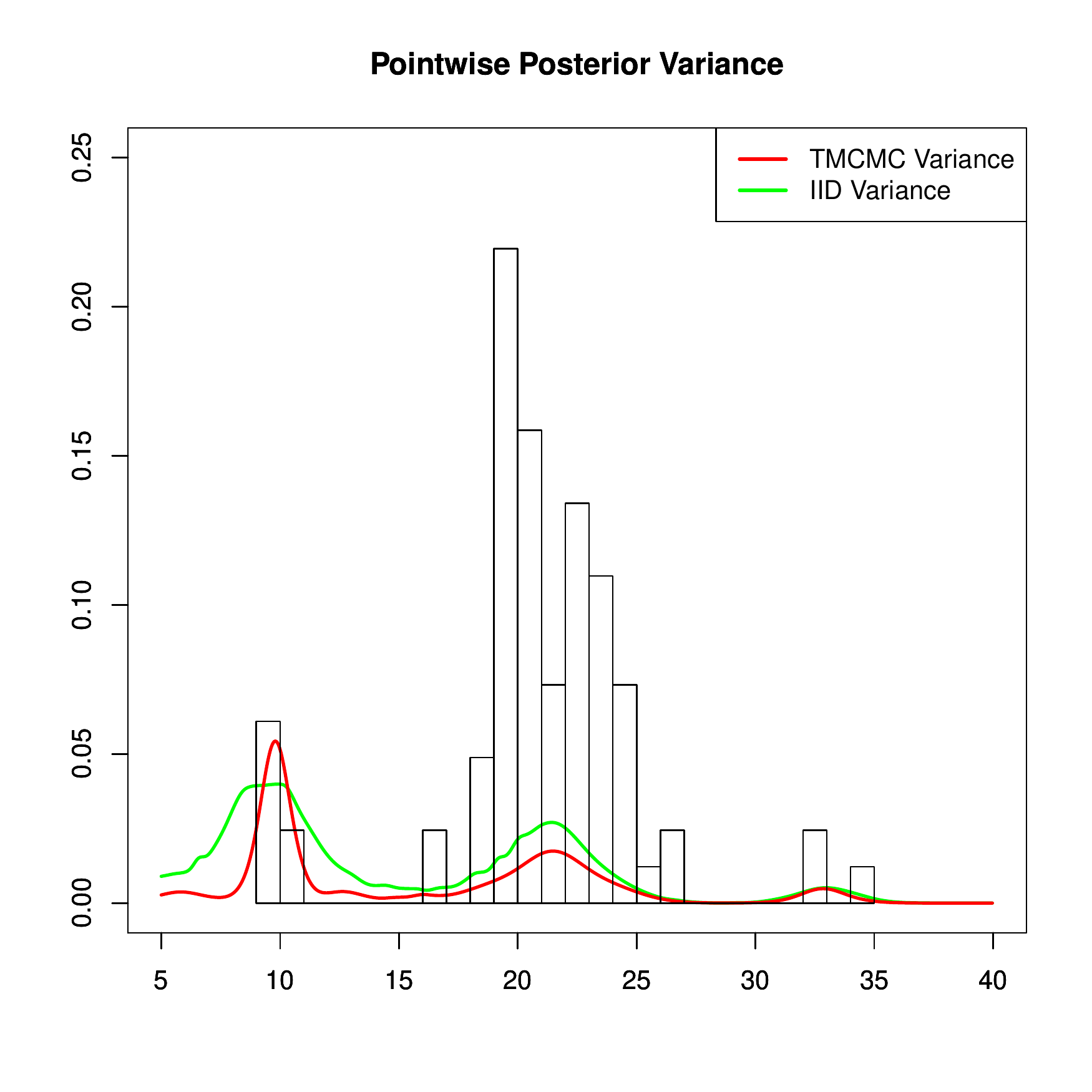}}\\
	\vspace{2mm}
	\subfigure [TMCMC autocorrelation plot for $\nu_{30}$.]{ \label{fig:galaxy3}
	\includegraphics[width=7.5cm,height=7.5cm]{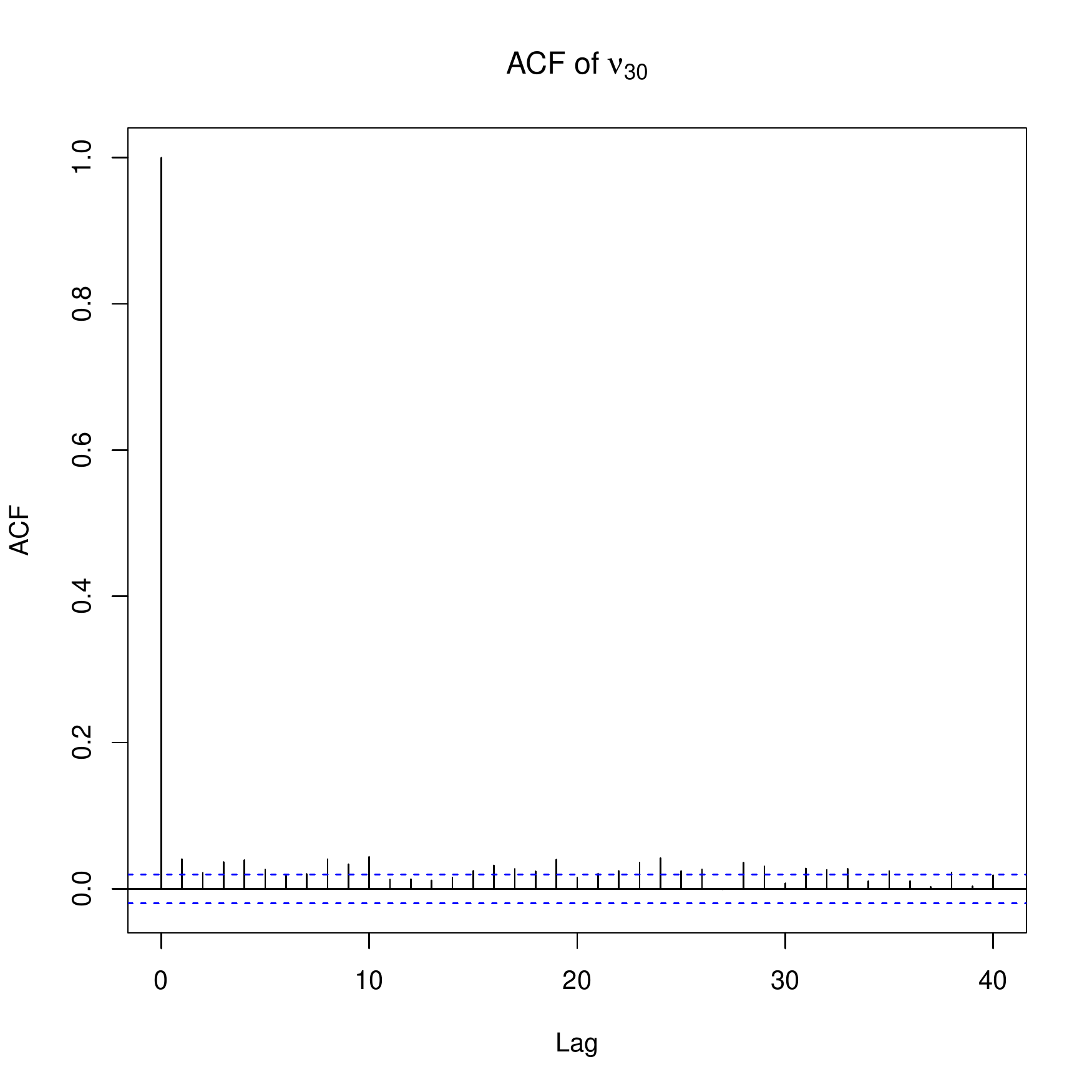}}
	\vspace{2mm}
	\subfigure [TMCMC autocorrelation plot for $\tau_{30}$.]{ \label{fig:galaxy4}
	\includegraphics[width=7.5cm,height=7.5cm]{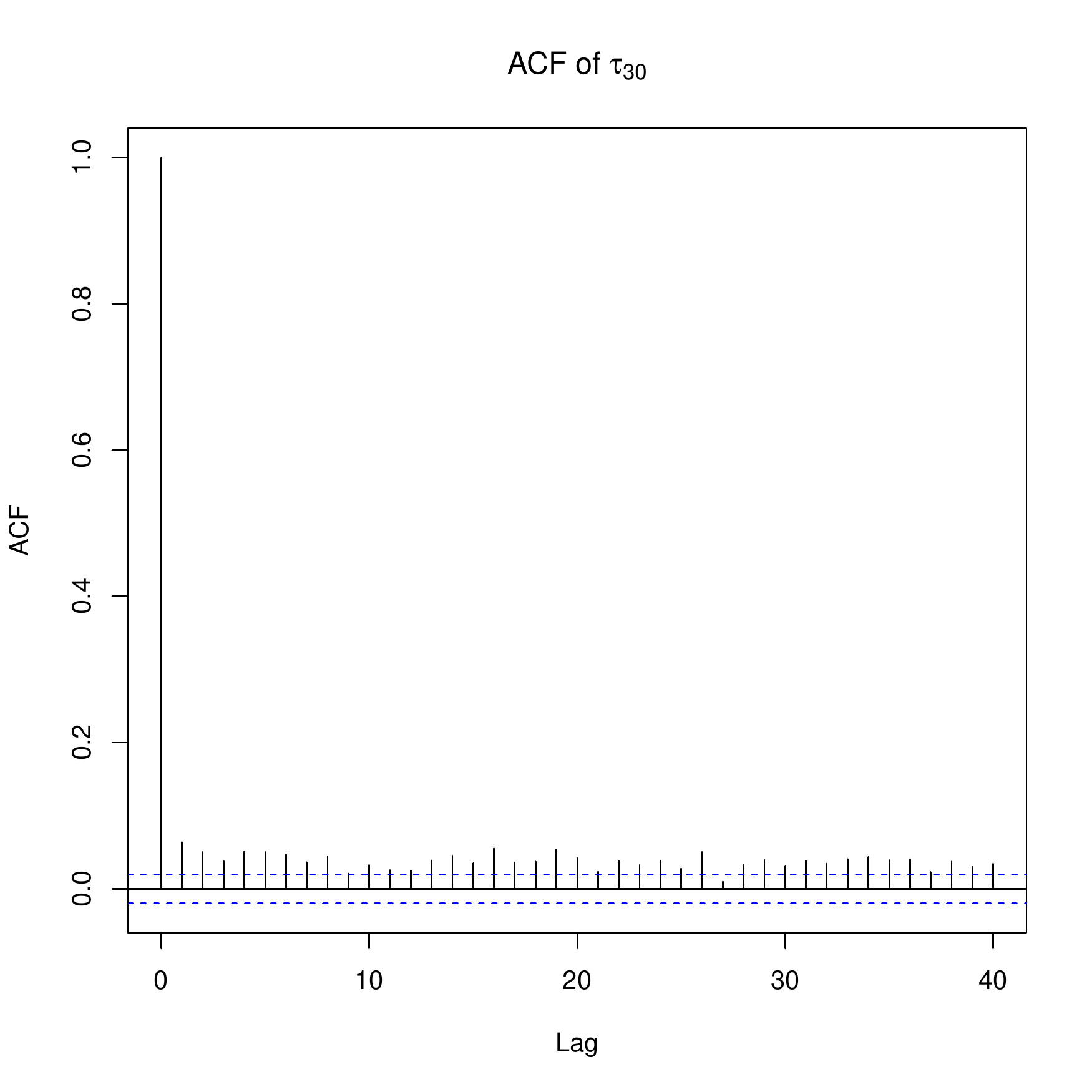}}
	\caption{IID and TMCMC sampling for the galaxy data.}
	\label{fig:galaxy}
\end{figure}

Now, we anticipate that there may arise the question that if smooth density estimators as shown in panel (a) of Figure \ref{fig:galaxy} reflect a model that
fails to capture the minor details of the histogram. Our response would be that the purpose of model-based analysis is to smooth the histogram, and capturing
minor details may be artifacts of the method employed to implement the model. To demonstrate, we implement our additive TMCMC algorithm once again for the galaxy data,
but now with the scaling constants set to $\sqrt{0.006}$. The corresponding TMCMC based density estimate, pointwise posterior predictive variances 
with respect to TMCMC and the autocorrelation plots are provided 
in Figure \ref{fig:galaxy_tmcmc}.
Panel (a) shows that the posterior predictive density based on the TMCMC realizations captures all the minor details of the histogram, and panel (b) shows
that the pointwise posterior predictive variances based on this TMCMC algorithm are much larger compared to panel (b) of Figure \ref{fig:galaxy}. Although the location parameters
do not exhibit substantial autocorrelations, as exemplified by panel (c), the scale parameters have high autocorrelations which refuse to die down even at lag $40$.

Such high autocorrelations are, in fact, responsible for the high pointwise posterior predictive variances of panel (b) and the deceptively accurate density estimate
of panel (a). The latter warrants further explanation. Note that high autocorrelation of $\tau_k$, for any $k=1,\ldots,30$, implies that the realizations of $\tau_k$ 
are not much different from each other. Hence, the correlation between $\tau_j$ and $\tau_k$, for $j\neq k$, will tend to be close to zero. This would effectively imply
many distinct $\tau_k$, which would enforce the same number of distinct $\nu_k$. The square roots of the inverse of these $\tau_k$ act as bandwidths for the density estimation,
and so there would be many distinct locations and the corresponding bandwidths. Together they reach out to every minor bump of the histogram and create the 
impression of great accuracy of the resultant density estimate. As we argued, such accuracy is nothing but an artifact of poor mixing of TMCMC taking small steps 
in each iteration, and hence must be considered spurious. Hence, Figure \ref{fig:galaxy} and not Figure \ref{fig:galaxy_tmcmc}, represents the correct Bayesian inference.
Also note that the posterior probability of the number of components $K=2$, $3$, $4$, $5$, $6$, $7$, $8$ here are $0.0021$, $0.0414$, $0.1988$, $0.3688$, $0.3020$, 
$0.0860$, $0.0009$, respectively.
Thus, this TMCMC algorithm supports more components than the correct $iid$ method or the efficient TMCMC method, which is in keeping with the above discussion
with respect to autocorrelations.
\begin{figure}
	\centering
	\subfigure [TMCMC-based posterior predictive density for galaxy.]{ \label{fig:galaxy_tmcmc1}
	\includegraphics[width=7.5cm,height=7.5cm]{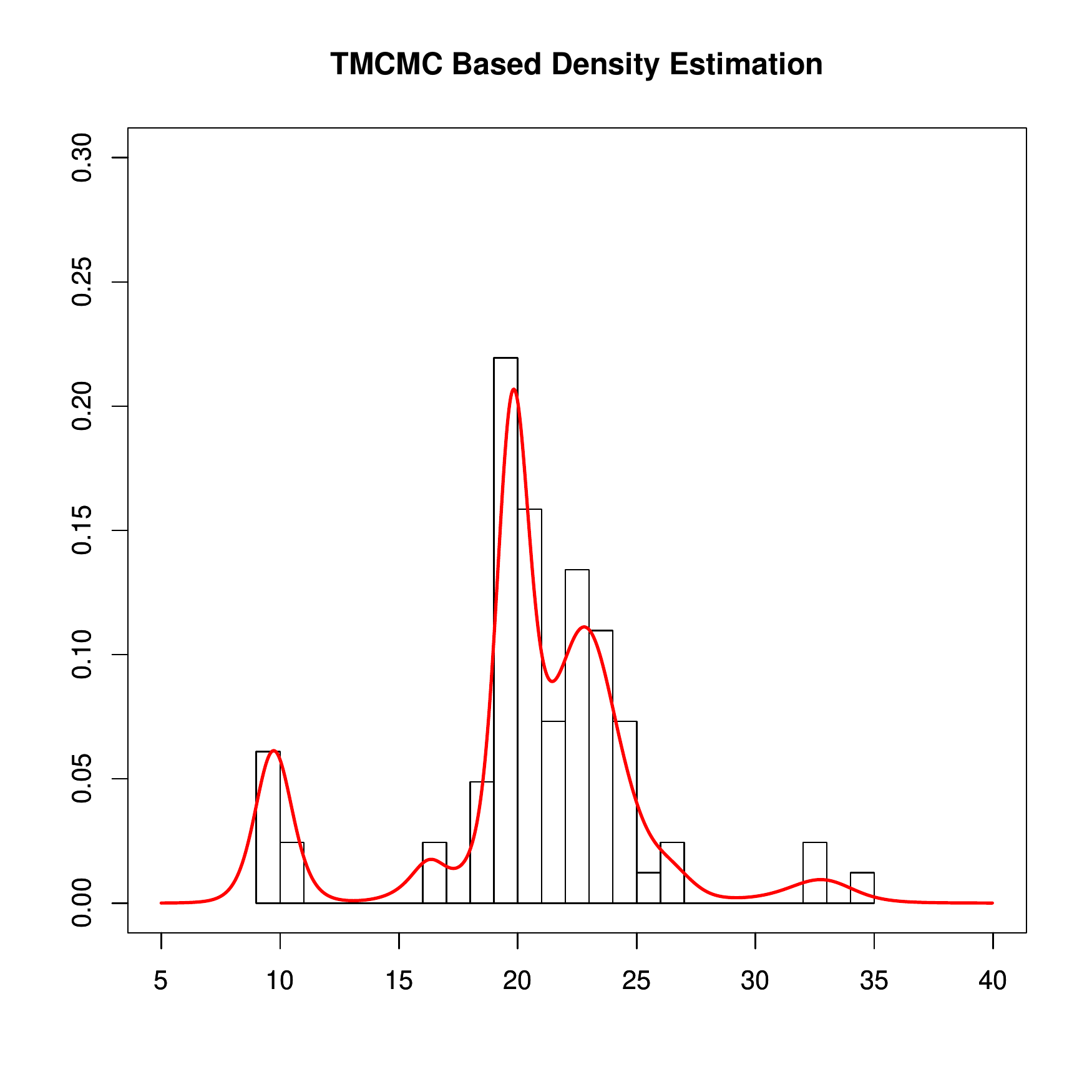}}
	\hspace{2mm}
	\subfigure [Poinwise variances with TMCMC.]{ \label{fig:galaxy_tmcmc2}
	\includegraphics[width=7.5cm,height=7.5cm]{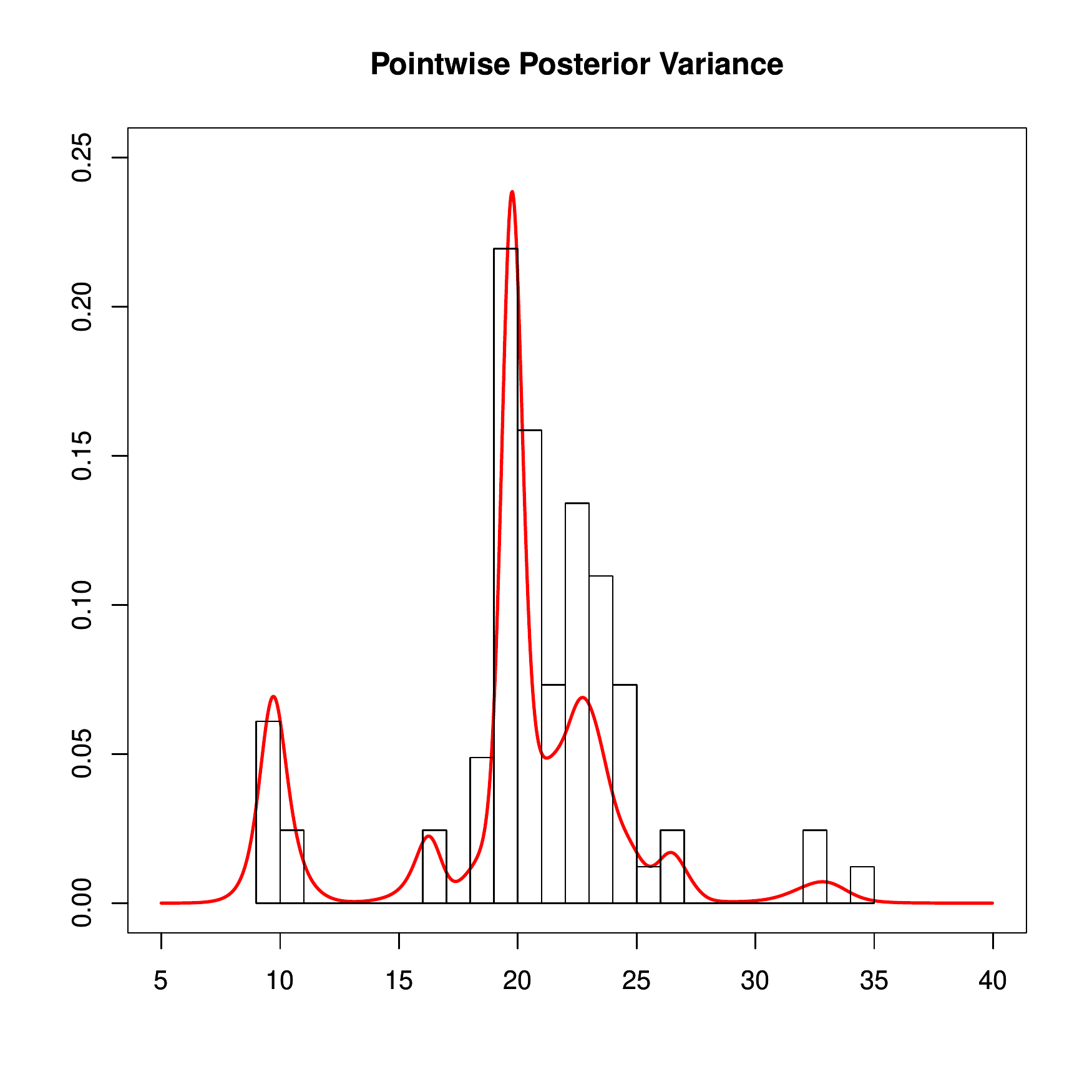}}\\
	\vspace{2mm}
	\subfigure [TMCMC autocorrelation plot for $\nu_{30}$.]{ \label{fig:galaxy_tmcmc3}
	\includegraphics[width=7.5cm,height=7.5cm]{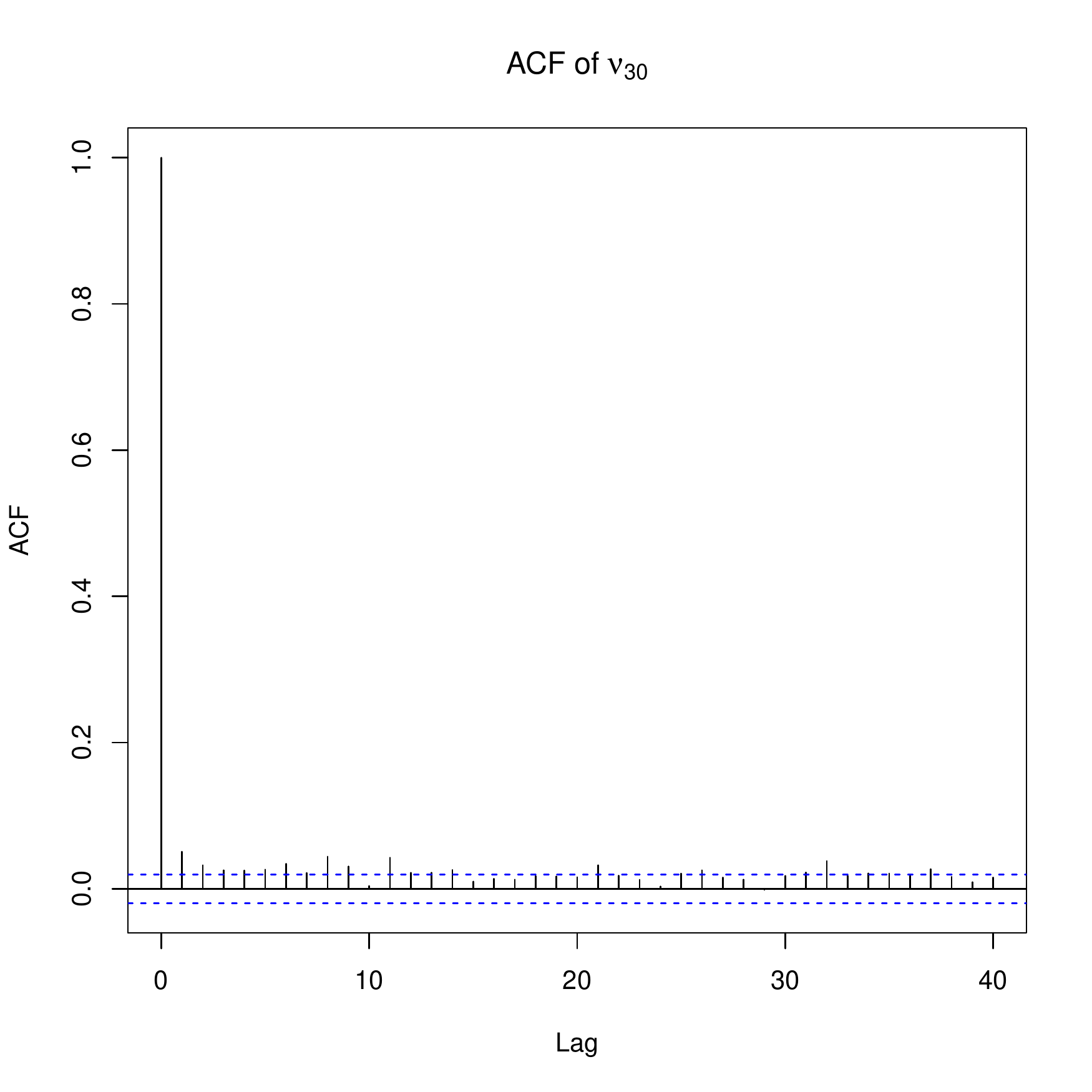}}
	\vspace{2mm}
	\subfigure [TMCMC autocorrelation plot for $\tau_{30}$.]{ \label{fig:galaxy_tmcmc4}
	\includegraphics[width=7.5cm,height=7.5cm]{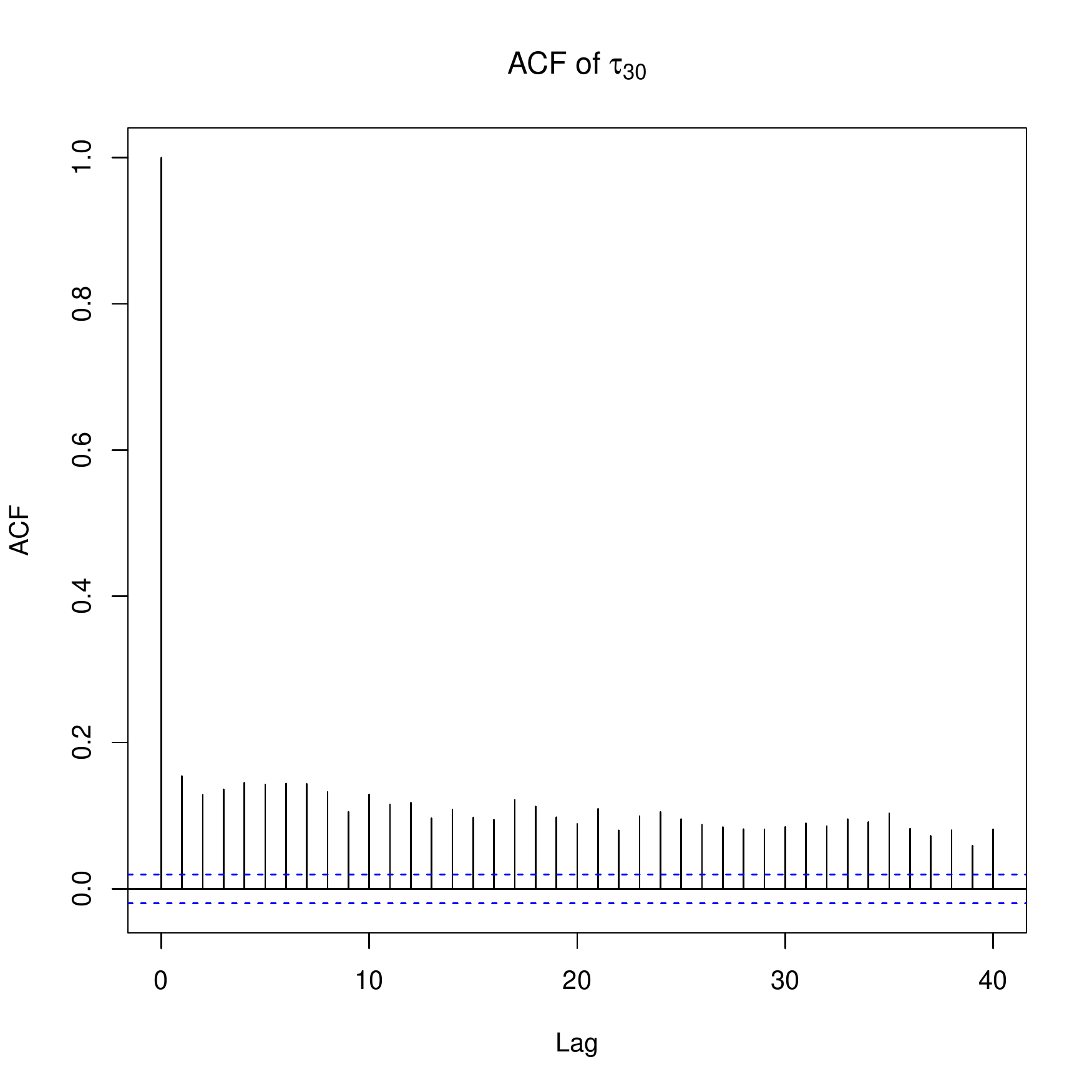}}
	\caption{TMCMC sampling for the galaxy data with small scaling constant.}
	\label{fig:galaxy_tmcmc}
\end{figure}

\section{Summary and conclusion}
\label{sec:conclusion}

MCMC sampling from posterior DP mixtures offers substantial challenges in terms of both mixing and implementation time. Despite the existence of a plethora of
MCMC algorithms for DP mixtures, it is extremely difficult to single out any algorithm for general application. More disconcertingly, it is not possible to rigorously
address if the underlying Markov chain has at all converged to the target DP mixture posterior. The ideal situation of $iid$ sampling is usually perceived as
inconceivable and impractical by the statistical and probabilistic community, even in finite-dimensional setups. In finite-dimensional situations, as well as in
multimodal and variable-dimensional contexts, and even for doubly intractable target distributions, we attempted to come up with efficient $iid$ sampling
procedures (\ctn{Bhatta21a}, \ctn{Bhatta21b}, \ctn{Bhatta21c}). In this article, we have attempted to provide a novel $iid$ sampling procedure for DP mixtures
in general, focussing particularly on the more general, flexible and efficient DP mixture model of \ctn{Bhattacharya08}. The key idea is of course a generalization
of our aforementioned works on $iid$ sampling, but the infinite-dimensional and discrete nature of DP called for some significant modification of our existing theory
to create a valid $iid$ sampling procedure for DP mixtures. Our theory does not depend upon conjugate or non-conjugate setups and works equally well for both
situations. Application of our $iid$ method to three benchmark datasets revealed excellent performance, including very fast parallel computation.

It is important to note that \ctn{Sabya12} had already created a novel perfect sampling procedure for the DP mixture of \ctn{Bhattacharya08}, integrating out
the random measure $G$ and creating appropriate bounding chains associated with an efficient Gibbs sampling procedure. The method encompasses both conjugate and
non-conjugate cases, and so, is highly relevant and comparable with our current work. However, the theory requires compact parameter space, which is not required
in this current work. Moreover, the computation required by \ctn{Sabya12} seems to be too intensive for generating a large number of $iid$ realizations. For instance,
application of their method to the galaxy data with $M=10$ took $11$ days to generate a single perfect realization! Parallelizing their method would only halve the time,
which still would not serve the purpose of generating adequate number of $iid$ realizations. In contrast, in our current work, we have been able to generate $10,000$
realizations for the galaxy data in just $5$ minutes, even with $M=30$! Although our procedure is based on truncating the random measure $G$, the upper bound
of Theorem \ref{theorem:truncation}, illustrated in detail in Remark \ref{remark:truncation}, 
shows almost indistinguishable agreement of the truncated model with the original one. Indeed, for all practical purposes, simulations from the original and the truncated
DP mixture models of \ctn{Bhattacharya08} would be identical. 

Although various advantages of \ctn{Bhattacharya08} over the traditional DP mixture model are established, Theorem \ref{theorem:truncation} and 
Remark \ref{remark:truncation} bring out yet another great advantage of the former with respect to truncation. Indeed, the truncated DP mixture model of \ctn{Bhattacharya08}
is in much closer agreement with the original one compared to that in the case of the traditional DP mixture model.

In fine, we remark that although DP mixtures clearly dominate the literature on Bayesian nonparametrics, there are various other classes of nonparametric Bayesian models
as well, for instance, those based on P\'{o}lya trees. In our future work, we intend to further generalize our $iid$ sampling procedure to encompass all
nonparametric Bayesian models.


\renewcommand\baselinestretch{1.3}
\normalsize
\bibliography{irmcmc}


\end{document}